\documentclass[journal,draftcls,onecolumn,12pt]{IEEEtran}

\usepackage{amsmath}
\usepackage{amssymb}
\usepackage[dvips]{graphicx}
\usepackage{curves,color}
\usepackage{subfigure}
\usepackage{dsfont,slashbox,wasysym}
\usepackage{tikz}
\usepackage{pdfpages,enumerate}
\usepackage{multirow}

\allowdisplaybreaks

\newtheorem{theorem}{Theorem}
\newtheorem{lemma}{Lemma}

\newtheorem{example}{Example}
\newtheorem{definition}{Definition}

\RequirePackage[normalem]{ulem} %DIF PREAMBLE
\RequirePackage{color}\definecolor{RED}{rgb}{1,0,0}\definecolor{BLUE}{rgb}{0,0,1} %DIF PREAMBLE
\begin{document}
\title{Design of Multiple-Edge Protographs for QC LDPC Codes Avoiding Short Inevitable Cycles}
\author{Hosung~Park, Seokbeom~Hong, Jong-Seon~No,~\IEEEmembership{Fellow,~IEEE,} and~Dong-Joon~Shin,~\IEEEmembership{Senior~Member,~IEEE}% 
%\thanks{Manuscript received June 18, 2011. This work was supported by the National Research Foundation of Korea (NRF) grant funded by the Korea government (MEST) (No.2012-0000186) and also supported by the KCC (Korea Communications Commission) under the R\&D program supervised by the KCA (Korea Communications Agency) (KCA-2012-08-911-04-003). The material in this paper was presented in part at the 2011 IEEE International Symposium on Information Theory.}%
\thanks{H.~Park, S.~Hong, and J.-S.~No are with the Department of Electrical Engineering and Computer Science, INMC, Seoul National University, Seoul 151-744, Korea (e-mail: lovepk98@snu.ac.kr, fousbyus@ccl.snu.ac.kr, jsno@snu.ac.kr).}%
\thanks{D.-J. Shin is with the Department of Electronic Engineering, Hanyang University, Seoul 133-791, Korea (e-mail: djshin@hanyang.ac.kr).}%
}%

%\markboth{IEEE TRANSACTIONS ON INFORMATION THEORY, VOL. XX, NO. XX, XXX 201X}
%{Park \MakeLowercase{\textit{et al.}}: Design of Multiple-Edge Protographs for QC LDPC Codes}

\maketitle

\begin{abstract}
There have been lots of efforts on the construction of quasi-cyclic (QC) low-density parity-check (LDPC) codes with large girth.
However, most of them are focused on protographs with single edges and little research has been done for the construction of QC LDPC codes lifted from protographs with multiple edges.
Compared to single-edge protographs, multiple-edge protographs have benefits such that QC LDPC codes lifted from them can potentially have larger minimum Hamming distance.
In this paper, all subgraph patterns of multiple-edge protographs, which prevent QC LDPC codes from having large girth by inducing inevitable cycles, are fully investigated based on graph-theoretic approach.
By using combinatorial designs, a systematic construction method of multiple-edge protographs is proposed for regular QC LDPC codes with girth at least 12 and also other method is proposed for regular QC LDPC codes with girth at least 14.
A construction algorithm of QC LDPC codes by lifting multiple-edge protographs is proposed and it is shown that the resulting QC LDPC codes have larger upper bounds on the minimum Hamming distance than those lifted from single-edge protographs.
Simulation results are provided to compare the performance of the proposed QC LDPC codes, the progressive edge-growth (PEG) LDPC codes, and the PEG QC LDPC codes.
\end{abstract}


\section{Introduction}\label{sec:introduction}
\IEEEPARstart{L}{ow}-density parity-check (LDPC) codes \cite{Gallager} have been one of major research topics in coding area over the past decade due to their near capacity-approaching performance. Since low decoding complexity can be achieved by various iterative decoding algorithms, LDPC codes have been adopted in many practical applications. Especially, quasi-cyclic (QC) LDPC codes are well suited for hardware implementation using simple shift registers due to the regularity in their parity-check matrices.

Thorpe \cite{Thorpe} introduced the concept of {\em protograph-based LDPC codes}, a class of LDPC codes lifted from protographs. QC LDPC codes belong to the protograph-based LDPC codes because they can be regarded as the lifted ones from the protographs using cyclic permutations. Therefore, the performance of QC LDPC codes mainly depends on how to design their protographs as well as how to assign shift values.

The performance of LDPC codes under message-passing algorithms depends on the girth of the codes because a message sent by a node along a cycle propagates back to the node itself after some iterations, which causes the dependence among messages and performance degradation. Therefore, there have been lots of efforts to construct QC LDPC codes with large girth \cite{Tanner}-\cite{Esmaeili}. In \cite{Fossorier}, necessary and sufficient conditions on determining the girth of QC LDPC codes from circulant permutation matrices are derived and some families of QC LDPC codes are constructed. Most of QC LDPC codes with large girth are constructed based on algebraic structures \cite{Tanner}-\cite{Milenkovic}, \cite{O'Sullivan}-\cite{Esmaeili} while some optimization algorithms and greedy search algorithms are used to find QC LDPC codes with large girth \cite{Wang}-\cite{Bocharova}. Various combinatorial designs have also been widely used to construct QC LDPC codes in order to guarantee the girth at least 6 \cite{Kim}-\cite{Vasic}.

The girth of QC LDPC codes constructed from protographs is determined by the structure of the protograph, the lift size, and all the shift values. In \cite{Tanner}, \cite{O'Sullivan}, \cite{Kim}, and \cite{Kelley}, an upper bound on the girth of QC LDPC codes, which is only determined by the structure of the protograph, is discussed.
Especially, in \cite{O'Sullivan}, all substructures of multiple-edge protographs, which inevitably give rise to cycles of length up to 12, are searched but any construction method of multiple-edge protographs for QC LDPC codes with large girth is not provided.
In \cite{Kim}, all substructures of single-edge protographs, which inevitably give rise to cycles of length up to 20 in QC LDPC codes, are identified and by using combinatorial designs, some single-edge protographs for girth larger than or equal to 18 and other single-edge protographs for girth larger than or equal to 14 are constructed.

Although the behavior of iterative message-passing decoders is mostly dominated by the pseudo-weight of pseudo-codewords \cite{Kelley1}, \cite{Koetter}, the minimum Hamming distance still plays an important role because it characterizes the undetectable errors and provides an upper bound on the minimum pseudo-weight of a code.
Smarandache and Vontobel \cite{Smarandache} derived two upper bounds on the minimum Hamming distance of QC LDPC codes, where one bound is applied when QC LDPC codes are explicitly given and the other bound can be applied even when only the protographs are given. It is shown by experiments that these upper bounds are very close to the real minimum Hamming distance when the lift size for a protograph is large enough. Also, through several examples, we can see that for the given protograph size and the given row- and column-weights, these two upper bounds  increase as the number of multiple edges increases in the protograph.
Therefore, these upper bounds can be increased if multiple-edge protographs are used to construct QC LDPC codes, compared to the case of single-edge protographs.

In this paper, multiple-edge protographs which can be lifted to QC LDPC codes with large girth are investigated. Search for all single- and multiple-edge subgraphs which inevitably generate cycles of any length in QC LDPC codes are systematically performed based on graph-theoretic approach as an extension of the results in \cite{O'Sullivan}, \cite{Kim}, and \cite{Kelley}. Construction methods of multiple-edge protographs using various combinatorial designs are proposed and a lifting algorithm to construct regular QC LDPC codes with large girth is also proposed.

The remainder of the paper is organized as follows.
Section \ref{sec:inevitable} introduces QC LDPC codes, protographs, and the concept of inevitable cycles.
In Section \ref{sec:subgraph}, all single- and multiple-edge subgraphs which generate inevitable cycles in QC LDPC codes are fully searched.
Based on these subgraph patterns, Section \ref{sec:12} describes a design method for multiple-edge protographs of regular QC LDPC codes having girth larger than or equal to 12.
In Section \ref{sec:14}, construction methods of multiple-edge protographs are proposed for regular QC LDPC codes having girth 14 when the variable node degree is 3 and they are generalized for regular QC LDPC codes with variable node degree larger than 3.
In Section \ref{sec:dmin}, a construction algorithm of QC LDPC codes lifted from the multiple-edge protographs is proposed. It is also shown that the proposed QC LDPC codes have larger upper bounds on the minimum Hamming distance than those lifted from single-edge protographs and the performance of the proposed QC LDPC codes is verified via numerical analysis.
Finally, the conclusions are provided in Section \ref{sec:conclusion}.

\vspace{2mm}
\section{Inevitable Cycles of QC LDPC Codes}\label{sec:inevitable}

\subsection{QC LDPC Codes}

Let $\mathcal{C}$ be a binary LDPC code whose parity-check matrix $H$ is a $J \times L$ array of $z \times z$ circulants or zero matrices as
\begin{equation*}
	H = \begin{bmatrix} H_{0,0}&H_{0,1}&\cdots&H_{0,L-1} \\ H_{1,0}&H_{1,1}&\cdots&H_{1,L-1}  \\ \vdots&\vdots&\ddots&\vdots \\ H_{J-1,0}&H_{J-1,1}&\cdots&H_{J-1,L-1} \end{bmatrix}
\end{equation*}
where a \textit{circulant} $H_{j,l}$ is defined as a matrix whose each row is a cyclic shift of the row above it. Such an LDPC code is called \textit{quasi-cyclic} because applying circular shifts to the length-$z$ subblocks of a codeword gives another codeword.
Also, a bipartite graph which has $H$ as its incidence matrix is called the \textit{Tanner graph} of $\mathcal{C}$.

The \textit{weight} of a circulant $H_{j,l}$ is defined as the number of nonzero elements in the first column and denoted by $\mathrm{wt}(H_{j,l})$.
A circulant is entirely described by the positions of nonzero elements in the first column.
Let $i$, $0 \leq i \leq z-1$, be the index of the $(i+1)$-st element in the first column.
Then, the \textit{shift value}(s) of a circulant is defined as the index (indices) of the nonzero element(s) in the first column. Note that a shift value takes the value from $0$ to $z-1$ and $\infty$ is used as a shift value of a zero matrix $H_{i,j}$.

QC LDPC codes can be fully represented by binary polynomials as shown in \cite{Smarandache}. This polynomial representation is based on the isomorphism between $z \times z$ binary circulants and the polynomial ring $\mathbb{F}_2[x]/(x^z+1)$. The \textit{polynomial parity-check matrix} $H(x)$ of $\mathcal{C}$ is defined as
\begin{equation*}
	H(x) = \begin{bmatrix} h_{0,0}(x)&h_{0,1}(x)&\cdots&h_{0,L-1}(x) \\ h_{1,0}(x)&h_{1,1}(x)&\cdots&h_{1,L-1}(x)  \\ \vdots&\vdots&\ddots&\vdots \\ h_{J-1,0}(x)&h_{J-1,1}(x)&\cdots&h_{J-1,L-1}(x) \end{bmatrix}
\end{equation*}
where $h_{j,l}(x) = \sum_{i=0}^{z-1} h_{j,l,i} x^i \in \mathbb{F}_2[x]/(x^z+1)$ and $h_{j,l,i}$ is the element with the index $i$ in the first column of $H_{j,l}$. We can see that the number of nonzero terms in $h_{j,l}(x)$, which is denoted by $\mathrm{wt}(h_{j,l}(x))$, is equal to $\mathrm{wt}(H_{j,l})$ and the degrees of all nonzero terms in $h_{j,l}(x)$ are equivalent to the shift values of $H_{j,l}$.

The \textit{protograph} \cite{Thorpe} of a QC LDPC code $\mathcal{C}$ is a bipartite graph whose incidence matrix is $P=[p_{j,l}]$, where $p_{j,l} = \mathrm{wt}(H_{j,l})$. There are two kinds of nodes in the protograph, where horizontal (check) nodes correspond to rows in $P$ and vertical (variable) nodes correspond to columns in $P$. The Tanner graph of $\mathcal{C}$ is constructed by copying the protograph $z$ times and cyclically permuting the same $z$ edges. Such copy-and-permute operation is called \textit{lifting} and the length of a subblock $z$ is also called the \textit{lift size} of $\mathcal{C}$.
If $p_{j,l}\geq 2$, there are multiple edges between the horizontal node with index $j$ and the vertical node with index $l$ in the protograph.
A shift value is assigned to each edge in the protograph so that an edge is lifted by using the cyclic permutation with the assigned shift value to generate $\mathcal{C}$.
Note that, in this paper, the term `protograph' refers to both the bipartite graph and its incidence matrix based on their equivalence.

\subsection{Inevitable Cycles}

Necessary and sufficient conditions on the existence of cycles in the Tanner graph of QC LDPC codes are derived in terms of shift values in \cite{Fossorier}. These conditions are only applied to single-edge protographs but they can be naturally extended to cover the case of multiple-edge protographs as in Lemma \ref{lemma:cycle}.

Let $G=(V,E)$ denote a graph with a set of vertices $V$ and a set of edges $E$. Let $v_k$ ($e_k$) represent a vertex (an edge) in $V$ ($E$). A \textit{walk} is an alternating sequence of vertices and edges, denoted by $v_{i_0} e_{i_0} v_{i_1} \cdots v_{i_{n-1}} e_{i_{n-1}} v_{i_n}$, where the vertices $v_{i_j}$ and $v_{i_{j+1}}$ are the endpoints of the edge $e_{i_j}$. The \textit{length of a walk} $W$, denoted by $l(W)$, is defined as the number of edges in $W$. A walk is \textit{closed} if $v_{i_n}=v_{i_0}$ and a walk is \textit{non-reversing} if $e_{i_j} \neq e_{i_{j+1}}$ for $j=0,1,\ldots,n-2$. A closed walk is said to be \textit{tailless} if $e_{i_{n-1}} \neq e_{i_0}$. In this paper, connected graphs are only considered and it is noted that a \textit{cycle} is defined as a closed walk whose traversed vertices and edges are all distinct and the length of the shortest cycle in a graph is called the \textit{girth} of the graph.

Cycles in the Tanner graph of a QC LDPC code are closely related to tailless non-reversing closed (TNC) walks in its protograph. The \textit{shift sum} of a walk $W$ in a protograph, denoted by $s(W)$, is defined as the alternating sum of shift values assigned to the edges in $W$, that is, $s(W) = \sum_{j=0}^{l(W)-1} (-1)^{j}(\mathrm{shift} ~ \mathrm{value} ~ \mathrm{of} ~ e_{i_j})$. Lemma \ref{lemma:cycle} shows necessary and sufficient conditions for a cycle of a certain length in the Tanner graph of QC LDPC codes to be generated from the protographs, which can be applied to both single-edge protographs and multiple-edge protographs. Its proof is directly derived from the results in \cite{Fossorier} and \cite{Kelley}.

\vspace{2mm}
\begin{lemma}
Let $\mathcal{W}$ denote the set of all TNC walks of length $n$ in a protograph. Suppose that a QC LDPC code is lifted from the protograph with the lift size $z$. Then, the Tanner graph of this QC LDPC code has a cycle of length $n$ if and only if there exists a walk $W \in \mathcal{W}$ such that $s(W)=0 ~\mathrm{mod}~ z$ and $W$ does not contain any shorter TNC walks with the zero shift sum.
\label{lemma:cycle}
\end{lemma}

\vspace{2mm}
The girth of QC LDPC codes is determined by the structure of the protograph, the lift size, and all the shift values assigned to edges. However, we can derive an upper bound on the girth of QC LDPC codes lifted from protographs without considering the lift size and the shift values based on the concept of inevitable cycles \cite{Tanner}, \cite{O'Sullivan}, \cite{Kim}.

\begin{figure}[tb]
	\centering
	\subfigure[An inevitable cycle of length 12.]{\includegraphics[scale=0.6]{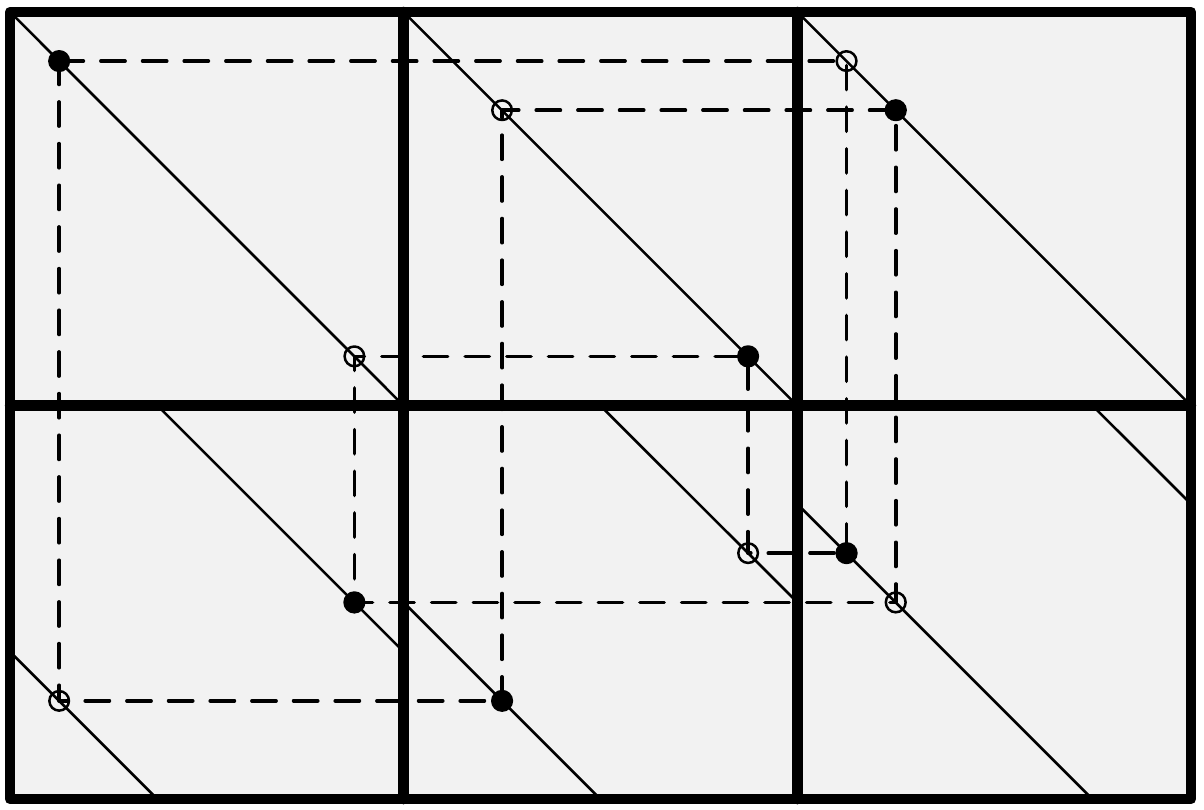}%
	\label{fig:inevitable_12}}\hspace{2mm}
	\subfigure[An inevitable cycle of length 10.]{\includegraphics[scale=0.6]{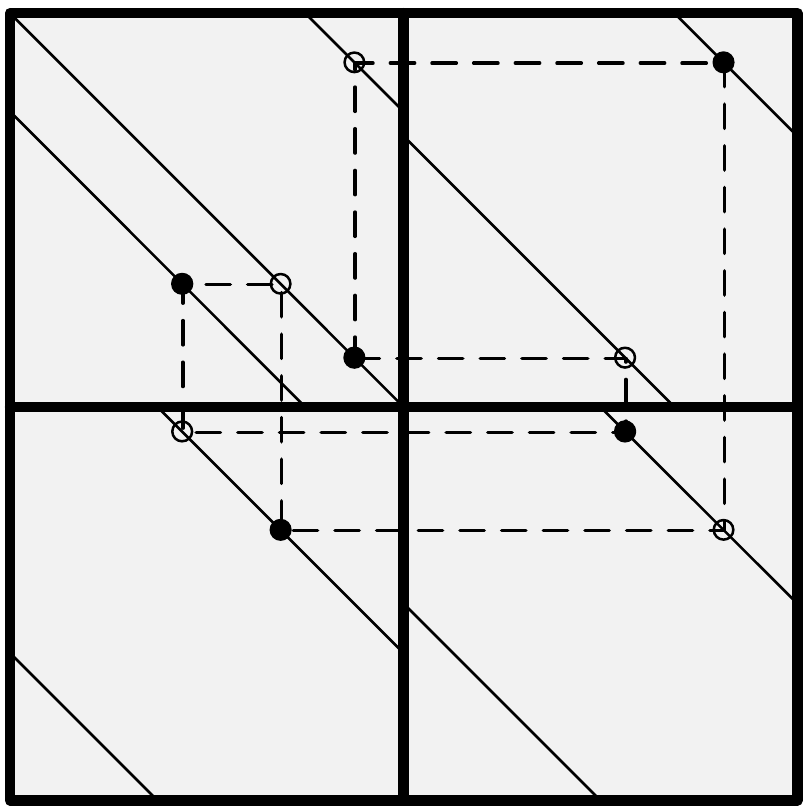}%
	\label{fig:inevitable_10}}
	\caption{Examples of inevitable cycles in QC LDPC codes.}
	\label{fig:inevitable}
\end{figure}

\vspace{2mm}
\begin{definition}
An \textit{inevitable cycle} induced by a protograph is defined as the cycle which always appears in the QC LDPC code lifted from the protograph regardless of the lift size and the shift values.
\end{definition}

\vspace{2mm}
It is well known that a QC LDPC code whose protograph has the $2 \times 3$ (or $3 \times 2$) all-one matrix as its submatrix must have the inevitable cycles of length 12 \cite{Tanner}, \cite{Fossorier}. In other words, the girth of this QC LDPC code is less than or equal to 12. Such an inevitable cycle of length 12 is depicted in Fig. \ref{fig:inevitable_12}.
Also, in QC LDPC codes lifted from multiple-edge protographs, inevitable cycles can be induced. As an example, Fig. \ref{fig:inevitable_10} shows an inevitable cycle of length 10, which appears in QC LDPC codes lifted from protographs with double edges. We can see that for a certain subgraph structure, inevitable cycles are always generated no matter what shift values are assigned to edges.

\vspace{2mm}
\section{Subgraphs of Multiple-Edge Protographs Inducing Inevitable Cycles}\label{sec:subgraph}

In order for QC LDPC codes to have large girth, their protographs should not contain the subgraphs which induce short inevitable cycles in the QC LDPC codes and thus it is necessary to find out all such subgraphs. From now on, the terms \textit{`an inevitable-cycle-inducing (ICI) subgraph of length $2i$'} will refer to a subgraph inducing inevitable cycles of length $2i$ . In \cite{O'Sullivan}, ICI subgraphs of length up to 12 in single- and multiple-edge protographs were fully investigated and, in \cite{Kim}, all ICI subgraphs of lengths 12 to 20 in single-edge protographs were searched by a brute force method. After that, a graph-theoretical framework was provided in \cite{Kelley}, which can be used to search all single- and multiple-edge ICI subgraphs. In this section, we will search and provide all ICI subgraphs as an extension of \cite{O'Sullivan}, \cite{Kim}, and \cite{Kelley}.

Define $\mathcal{P}_{2i}$ as a set of all irreducible ICI subgraphs of length $2i$ satisfying the following conditions:
\begin{enumerate}
	\item A subgraph $P \in \mathcal{P}_{2i}$ induces inevitable cycles of length $2i$ in the QC LDPC code.
	\item A subgraph $P \in \mathcal{P}_{2i}$ does not contain any proper subgraph which induces inevitable cycles of length less than or equal to $2i$.
	\item The number of rows in a subgraph $P \in \mathcal{P}_{2i}$ is not larger than that of columns.
	\item From each isomorphic class in $\mathcal{P}_{2i}$, only one protograph must be chosen as a representative of that class.
\end{enumerate}
The conditions 1) and 2) guarantee that if a protograph does not have any subgraph $P \in \mathcal{P}_{2i'}$ for $i' < i$, the QC LDPC code appropriately lifted from this protograph has girth larger than or equal to $2i$. A subgraph $P \in \mathcal{P}_{2i}$ takes an irreducible form because the condition 2) implies that if any edge is removed from $P$, it cannot induce inevitable cycles of length $2i$. Conditions 3) and 4) are required to choose a unique representative for each isomorphic class of subgraphs inducing inevitable cycles of length $2i$.

For identifying $\mathcal{P}_{2i}$, we need to investigate the relationship between inevitable cycles and TNC walks. A TNC walk $W$ is called \textit{abelian-forcing} \cite{Kelley} if for each edge in $W$, the number of traversals of the edge in a direction is the same as that in the opposite direction. Clearly, the shift sum of abelian-forcing TNC walks is zero regardless of the shift values of their edges. An abelian-forcing TNC walk is said to be \textit{simple} if it does not contain any shorter abelian-forcing TNC walks. It is obvious that inevitable cycles of QC LDPC codes are generated from simple abelian-forcing TNC (SAFTNC) walks in protographs.

\vspace{2mm}
\begin{lemma}
Any abelian-forcing TNC walk contains at least two different cycles.
\label{lemma:two_cycles}
\end{lemma}
\begin{proof}
Consider an abelian-forcing TNC walk $W=v_{i_0} e_{i_0} v_{i_1} \cdots v_{i_{n-1}} e_{i_{n-1}} v_{i_n}$. There exists a vertex $v_j$ such that $v_{i_{k}}=v_{i_{l}}=v_j$ for some $k \neq l$. Also, there exists a path $v_m e_{i_{p-1}} v_{i_{p}} \cdots v_{i_{q}} e_{i_{q}} v_m$ in $W$ such that all vertices from $v_{i_p}$ to $v_{i_q}$ are distinct. Since $W$ is non-reversing and tailless, that path forms a cycle and thus $W$ contains at least one cycle.

Assume that $W$ contains only one cycle.
Since $W$ is abelian-forcing, there exists a path $v_f e_g v_h e_{i_{a-1}} v_{i_a} \cdots v_{i_{b}} e_{i_{b}} v_h e_g v_f$ in $W$ such that $v_{i_j} \neq v_h$ for all $j=a,a+1,\ldots,b$. This contradicts the assumption of $W$ because $W$ cannot move from a vertex to itself without reversing. Therefore, $W$ contains at least two different cycles.
\end{proof}

\vspace{2mm}
As in \cite{Kelley}, two classes of graphs are defined as illustrated in Fig. \ref{fig:graph}.

\vspace{2mm}
\begin{definition}[\cite{Kelley}]
A $(x_1,x_2,x_3)$-\textit{theta graph}, denoted by $T(x_1,x_2,x_3)$, is a graph consisting of two vertices, each of degree three, that are connected to each other via three disjoint paths $X_1$, $X_2$, $X_3$ of the number of edges $x_1 \geq 1$, $x_2 \geq 1$, and $x_3 \geq 1$, respectively. A $(z_1,z_2;y)$-\textit{dumbbell graph}, denoted by $D(z_1,z_2;y)$, is a connected graph consisting of two edge-disjoint cycles $Z_1$ and $Z_2$ of the number of edges $z_1 \geq 1$ and $z_2 \geq 1$, respectively, that are connected by a path $Y$ of the number of edges $y \geq 0$.
\end{definition}

\begin{figure}[tb]
	\centering
	\subfigure[Theta graph.]{\includegraphics[scale=0.65]{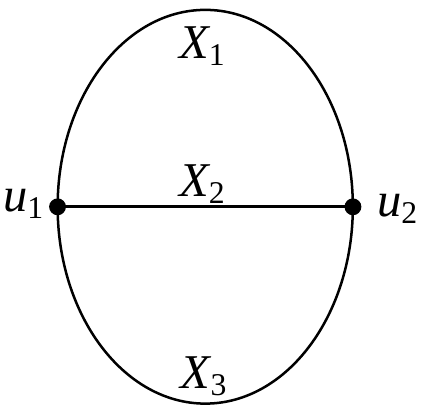}%
	\label{fig:theta}} \hspace{5mm}
	\subfigure[Dumbbell graph.]{\includegraphics[scale=0.65]{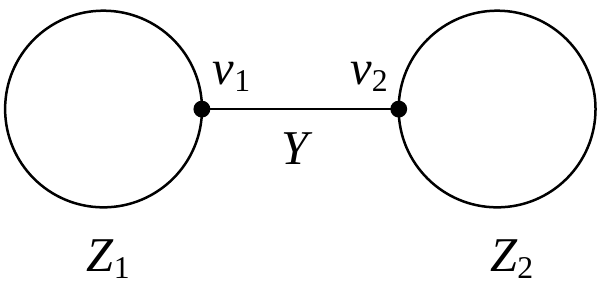}%
	\label{fig:dumbbell}}
	\caption{Theta graph and dumbbell graph.}
	\label{fig:graph}
\end{figure}

\vspace{2mm}
\begin{lemma}
Connecting two different cycles always results in either a theta graph or a dumbbell graph.
\label{lemma:connecting_two}
\end{lemma}
\begin{proof}
Let $C_1$ and $C_2$ denote two different cycles. Then, $C_1$ and $C_2$ can be connected in only three ways: The number of common vertices in $C_1$ and $C_2$ is (i) 0, (ii) 1, or (iii) larger than or equal to 2. For the cases (i) and (ii), $C_1$ and $C_2$ form $D(z_1,z_2;y)$ with $y>0$ or $y=0$, respectively. In the case (iii), $T(x_1,x_2,x_3)$ is formed where $C_1 = X_1 \bigcup X_2$, $C_2 = X_2 \bigcup X_3$, and $x_2+1$ is the number of the common vertices.
\end{proof}

\vspace{2mm}
\begin{lemma}
The lengths of SAFTNC walks in $T(x_1,x_2,x_3)$ and $D(z_1,z_2;y)$ are $2(x_1+x_2+x_3)$ and $2(z_1+z_2)+4y$, respectively.
\label{lemma:length}
\end{lemma}
\begin{proof}
Consider $T(x_1,x_2,x_3)$ in Fig. \ref{fig:theta}. Let $u_1$ and $u_2$ denote the left and the right vertices of degree three, respectively, and let $X_1$, $X_2$, and $X_3$ be the paths from $u_1$ to $u_2$. Also, let $\bar{X}_1$, $\bar{X}_2$, and $\bar{X}_3$ denote the reverse paths of $X_1$, $X_2$, and $X_3$, respectively. Then we can see that an SAFTNC walk $X_1 \bar{X}_2 X_3 \bar{X}_1 X_2 \bar{X}_3$ has the length $2(x_1+x_2+x_3)$ and any other SAFTNC walks possibly generated in $T(x_1,x_2,x_3)$ have the same length.

Similarly, consider $D(z_1,z_2;y)$ in Fig. \ref{fig:dumbbell}. Let $v_1$ and $v_2$ denote the left and the right vertices of degree three, respectively, and let $Z_1$ and $Z_2$ be the cycles rotating clockwise from $v_1$ and $v_2$, respectively, and let $Y$ be the path from $v_1$ to $v_2$. Also, let $\bar{Z}_1$, $\bar{Z}_2$, and $\bar{Y}$ denote the reverse paths of $Z_1$, $Z_2$, and $Y$, respectively. Then we can see that an SAFTNC walk $Z_1 Y Z_2 \bar{Y} \bar{Z}_1 Y \bar{Z}_2 \bar{Y}$ has the length $2(z_1+z_2)+4y$ and any other SAFTNC walks possibly generated in $D(z_1,z_2;y)$ have the same length.
\end{proof}

\vspace{2mm}
Note that if any edge is removed from $T(x_1,x_2,x_3)$ or $D(z_1,z_2;y)$, those inherent SAFTNC walks disappear and thus $T(x_1,x_2,x_3)$ and $D(z_1,z_2;y)$ are of irreducible form. Now, we will check whether it is sufficient to only consider theta graphs and dumbbell graphs for $\mathcal{P}_{2i}$.

\vspace{2mm}
\begin{lemma}
Suppose that a graph $G$ contains at least one of theta graphs or dumbbell graphs as its proper subgraphs. The shortest SAFTNC walk in $G$ occurs only in a theta graph or a dumbbell graph.
\label{lemma:connecting_three}
\end{lemma}
\begin{proof}
Let $W$ denote the shortest SAFTNC walk and assume that $W$ traverses all edges in $G$.
From Lemmas \ref{lemma:two_cycles} and \ref{lemma:connecting_two}, $W$ should contain a theta graph or a dumbbell graph.
Consider the following two cases: (i) $G$ has some theta graphs, (ii) $G$ does not have any theta graphs.

In the case (i), we first note that $l(W)$ is at least twice the number of edges in $G$ due to the definition of abelian-forcing TNC walks. The SAFTNC walk only generated by a theta graph in $G$ is shorter than $W$ because the SAFTNC walk has the length exactly twice the number of edges in the theta graph. This contradicts the assumption that $W$ is the shortest one. In the case (ii), we note that a simple abelian-forcing TNC walk should traverse the edge not belonging to any cycles at least four times because if the walk traverses the edge twice, the walk will include two simple abelian-forcing TNC walks each of which occurs in the different side of the edge. Since cycles in G are connected with each other via only one path which does not belong to any cycles, the SAFTNC walk only generated by a dumbbell graph in $G$ is shorter than $W$. This contradicts the assumption that $W$ is the shortest one. Therefore, $W$ occurs only in a theta graph or a dumbbell graph.
\end{proof}

\vspace{2mm}                                    
In the next theorem, $\mathcal{P}_{2i}$ will be identified.

\vspace{2mm}
\begin{theorem}
$\mathcal{P}_{2i}$ is a collection of all $T(x_1,x_2,x_3)$'s with $2(x_1+x_2+x_3)=2i$ and all $D(z_1,z_2;y)$'s with $2(z_1+z_2)+4y=2i$.
\label{thm:ICI}
\end{theorem}
\begin{proof}
From Lemmas \ref{lemma:two_cycles} and \ref{lemma:connecting_three}, any subgraph $P \in \mathcal{P}_{2i}$ should be either a theta graph or a dumbbell graph. Therefore, the proof is completed by Lemma \ref{lemma:length}.
\end{proof}

\vspace{2mm}
Now we can find all single- and multiple-edge ICI subgraphs from $T(x_1,x_2,x_3)$ and $D(z_1,z_2;y)$. A representative of an isomorphic class in $\mathcal{P}_{2i}$ can be uniquely chosen by selecting parameters satisfying the following conditions:
\begin{itemize}
	\item $x_1 \geq x_2 \geq x_3 \geq 1$
	\item $x_1$, $x_2$, $x_3$ are all even or all odd
	\item $z_1 \geq z_2 \geq 2$, $y \geq 0$
	\item $z_1$ and $z_2$ are even.
\end{itemize}
Note that the second and the fourth conditions are derived because each subgraph $P \in \mathcal{P}_{2i}$ is a bipartite graph.

\begin{table}[!t]
	\renewcommand{\arraystretch}{1.3}
	\caption{All ICI Subgraphs of Length up to 20 (T: Theta Graph, D: Dumbbell Graph, S: Single-Edge, M: Multiple-Edge)}
	\label{table:ICI}
	\centering
	\begin{tabular}{|c|c||c||c||c||c|c||c|c||c|c|c|c||c|c|c||c|c|c|c|c|c|}
		\hline
		\multicolumn{2}{|c||}{$\mathcal{P}_{2i}$} & $\mathcal{P}_{6}$ & $\mathcal{P}_{8}$ & $\mathcal{P}_{10}$ & \multicolumn{2}{c||}{$\mathcal{P}_{12}$} & \multicolumn{2}{c||}{$\mathcal{P}_{14}$} & \multicolumn{4}{c||}{$\mathcal{P}_{16}$} & \multicolumn{3}{c||}{$\mathcal{P}_{18}$} & \multicolumn{6}{c|}{$\mathcal{P}_{20}$} \\
		\hline
		\hline
		\multirow{3}{*}{$T$} & $x_1$ & 1 & - & 3 & 2 & - & 3 & 5 & 4 & - & - & - & 3 & 5 & 7 & 4 & 6 & - & - & - & - \\ \cline{2-22}
		& $x_2$ & 1 & - & 1 & 2 & - & 3 & 1 & 2 & - & - & - & 3 & 3 & 1 & 4 & 2 & - & - & - & - \\ \cline{2-22}
		& $x_3$ & 1 & - & 1 & 2 & - & 1 & 1 & 2 & - & - & - & 3 & 1 & 1 & 2 & 2 & - & - & - & - \\
		\hline
		\multicolumn{2}{|c||}{Type} & M & - & M & S & - & S & M & S & - & - & - & S & S & M & S & S & - & - & - & - \\
		\hline
		\hline
		\multirow{3}{*}{$D$} & $z_1$ & - & 2 & - & 2 & 4 & - & - & 2 & 4 & 4 & 6 & - & - & -  & 2 & 4 & 4 & 6 & 6 & 8 \\ \cline{2-22}
		& $z_2$ & - & 2 & - & 2 & 2 & - & - & 2 & 2 & 4 & 2 & - & - & - & 2 & 2 & 4 & 2 & 4 & 2 \\ \cline{2-22}
		& $y$ & - & 0 & - & 1 & 0 & - & - & 2 & 1 & 0 & 0 & - & - & - & 3 & 2 & 1 & 1 & 0 & 0 \\
		\hline
		\multicolumn{2}{|c||}{Type} & - & M & - & M & M & - & - & M & M & S & M & - & - & - & M & M & S & M & S & M \\
		\hline
	\end{tabular}
\end{table}

%\begin{table}[!t]
%	\renewcommand{\arraystretch}{1.3}
%	\caption{All ICI Subgraphs of Length up to 18 (M: Multiple-Edge Subgraph, S: Single-Edge Subgraph)}
%	\label{table:ICI}
%	\centering
%	\begin{tabular}{|c||c|c|c||c|c|c|}
%		\hline
%		\multirow{2}{*}{$\mathcal{P}_{2i}$} & \multicolumn{3}{c||}{$T(x_1,x_2,x_3)$} & \multicolumn{3}{c|}{$D(z_1,z_2;y)$} \\ \cline{2-7}
%		& $x_1$ & $x_2$ & $x_3$ & $z_1$ & $z_2$ & $y$ \\
%		\hline
%		\hline
%		$\mathcal{P}_{6}$ & 1 & 1 & 1 & - & - & - \\
%		\hline
%		$\mathcal{P}_{8}$ & - & - & - & 2 & 2 & 0 \\
%		\hline
%		$\mathcal{P}_{10}$ & 3 & 1 & 1 & - & - & - \\
%		\hline
%		\multirow{2}{*}{$\mathcal{P}_{12}$} & 2 & 2 & 2 & 2 & 2 & 1 \\ \cline{2-7}
%		& - & - & - & 4 & 2 & 0 \\
%		\hline
%		\multirow{2}{*}{$\mathcal{P}_{14}$} & 5 & 1 & 1 & - & - & - \\ \cline{2-7}
%		& 3 & 3 & 1 & - & - & - \\
%		\hline
%		\multirow{4}{*}{$\mathcal{P}_{16}$} & 4 & 2 & 2 & 2 & 2 & 2 \\ \cline{2-7}
%		& - & - & - & 4 & 2 & 1 \\ \cline{2-7}
%		& - & - & - & 4 & 4 & 0 \\ \cline{2-7}
%		& - & - & - & 6 & 2 & 0 \\
%		\hline
%		\multirow{3}{*}{$\mathcal{P}_{18}$} & 3 & 3 & 3 & - & - & - \\ \cline{2-7}
%		& 5 & 3 & 1 & - & - & - \\ \cline{2-7}
%		& 7 & 1 & 1 & - & - & - \\ 
%		\hline
%	\end{tabular}
%\end{table}

According to Theorem \ref{thm:ICI}, each integer solution of the equations $2(x_1 + x_2 + x_3 )=2i$ and $2(z_1 + z_2 ) +4y = 2i$ forms one ICI subgraph in $\mathcal{P}_{2i}$. Note that all ICI subgraphs of any length can be easily found and $T(x_1,1,1)$ and $D(z_1,2;y)$ are ICI subgraphs having multiple edges. All ICI subgraphs of length up to 20 are listed as a form of theta or dumbbell graphs in Table \ref{table:ICI} and all ICI subgraphs of length up to $14$ are listed as a form of incidence matrices as follows:

\vspace{2mm}
$\mathcal{P}_6 ~= [3]$;

\vspace{2mm}
$\mathcal{P}_8 ~= \begin{bmatrix} 2 & 2 \end{bmatrix}$;

\vspace{2mm}
$\mathcal{P}_{10}= \begin{bmatrix} 2 & 1 \\ 1 & 1 \end{bmatrix}$;

\vspace{2mm}
$\mathcal{P}_{12}= \begin{bmatrix} 2 & 1 \\ 0 & 2 \end{bmatrix},
\begin{bmatrix} 2 & 1 & 1 \\ 0 & 1 & 1 \end{bmatrix},
\begin{bmatrix} 1 & 1 & 1 \\ 1 & 1 & 1 \end{bmatrix}$;

\vspace{2mm}
$\mathcal{P}_{14}= \begin{bmatrix} 2 & 1 & 0 \\ 1 & 0 & 1 \\ 0 & 1 & 1 \end{bmatrix},
\begin{bmatrix} 1 & 1 & 1 \\ 1 & 1 & 0 \\ 1 & 0 & 1 \end{bmatrix}$

%\vspace{2mm}
%$\mathcal{P}_{16}= \begin{bmatrix} 2 & 1 & 0 \\ 0 & 1 & 2 \end{bmatrix},
%\begin{bmatrix} 2 & 1 & 0 \\ 0 & 1 & 1 \\ 0 & 1 & 1 \end{bmatrix},
%\begin{bmatrix} 2 & 1 & 0 & 1 \\ 0 & 1 & 1 & 0 \\ 0 & 0 & 1 & 1 \end{bmatrix},\begin{bmatrix} 1 & 1 & 0 & 0 \\ 1 & 1 & 1 & 1 \\ 0 & 0 & 1 & 1 \end{bmatrix},
%\begin{bmatrix} 1 & 1 & 1 & 0 \\ 1 & 1 & 0 & 1 \\ 0 & 0 & 1 & 1 \end{bmatrix}$;
%
%\vspace{2mm}
%$\mathcal{P}_{18}= \begin{bmatrix} 2 & 1 & 0 & 0 \\ 1 & 0 & 1 & 0 \\ 0 & 1 & 0 & 1 \\ 0 & 0 & 1 & 1 \end{bmatrix},
%\begin{bmatrix} 1 & 1 & 0 & 0 \\ 1 & 1 & 1 & 0 \\ 0 & 0 & 1 & 1 \\ 0 & 1 & 0 & 1 \end{bmatrix},
%\begin{bmatrix} 1 & 1 & 0 & 0 \\ 0 & 1 & 1 & 0 \\ 1 & 0 & 1 & 1 \\ 0 & 1 & 0 & 1 \end{bmatrix}$

\vspace{4mm}\hspace{-5.5mm}
where the single-edge ICI subgraphs in Table \ref{table:ICI} were also listed in \cite{Kim} and all ICI subgraphs of length up to 12 were also listed in \cite{O'Sullivan}. Note that the transpose of each ICI subgraph also generates inevitable cycles of the same length and thus $\mathcal{P}_{2i}$ will be used to denote both the listed matrices and their transposes.

\vspace{2mm}
\section{Construction of Regular Protographs Avoiding Inevitable Cycles of Length Less Than 12}\label{sec:12}

In this section, we will construct regular protographs which avoid inevitable cycles of length less than 12 in QC LDPC codes. Consider a regular $J \times L$ protograph of which the column- and row-weights are $d_v$ and $d_c$, respectively, where $J < L$. If triple or more edges exist in the protograph, the girth of the lifted QC LDPC code is limited to 6 because of $\mathcal{P}_6 = [3]$. Therefore, only protographs with single and double edges will be considered in this paper. Let $n_2$ denote the number of double edges in the protograph.

Most of the considered protographs have at least two cycles and thus they always induce some inevitable cycles according to Lemmas \ref{lemma:connecting_two} and \ref{lemma:length}. Note that even if a protograph is designed not to contain any $\mathcal{P}_{2i'}$ with $i'<i$ so that inevitable cycles of length less than $2i$ are avoided, this protograph may have some inevitable cycles of length larger than or equal to $2i$.

To construct protographs which do not induce inevitable cycles of length less than 10, a pair of $2$'s should not appear in any row or in any column of the protograph to avoid $\mathcal{P}_8$. As in the next lemma, the number of double edges in a protograph should be upper bounded by the number of horizontal nodes to construct QC LDPC codes with girth larger than or equal to 10.

\vspace{2mm}
\begin{lemma}
If a $J \times L$ protograph does not induce inevitable cycles of length less than 10, then $n_2 \leq J$. 
\label{lemma:1}
\end{lemma}
\begin{proof}
If $n_2 > J$, there always exists a row which has at least two 2's and thus the protograph contains $\mathcal{P}_8$. This contracts the assumption.
\end{proof}

\vspace{2mm}
In order for QC LDPC codes to have the girth larger than or equal to 12, their protographs should not contain $\mathcal{P}_6$, $\mathcal{P}_8$, and $\mathcal{P}_{10}$. We will explain that an incidence matrix of a \textit{balanced ternary design} (BTD) with $\rho_2 = 1$ and $\lambda = 2$ is also the incidence matrix of a regular protograph with $n_2=J$ that does not induce inevitable cycles of length less than 12.

\vspace{2mm}
\begin{definition}[\cite{Colbourn1}]
A \textit{balanced ternary design} BTD$(v,b;\rho_1,\rho_2,r;k,\lambda)$ is an arrangement of $v$ elements $\{ 1,2,\ldots,v\}$ into $b$ multisets, or blocks, each of cardinality $k$, $k\leq v$, satisfying that (i) each element appears $r=\rho_1+2\rho_2$ times altogether, with multiplicity one in exactly $\rho_1$ blocks, with multiplicity two in exactly $\rho_2$ blocks and (ii) every pair of distinct elements appears $\lambda$ times, i.e., if $m_{j,h}$ is the multiplicity of the element $j$ in the $h$-th block, then for any elements $i$ and $j$ with $i \neq j$, we have $\sum_{h=1}^{b}m_{i,h}m_{j,h}=\lambda$.
\end{definition}

\vspace{2mm}
Note that a $v \times b$ incidence matrix of a BTD$(v,b;\rho_1,\rho_2,r;k,\lambda)$ is simply expressed as $[m_{j,h}]$ and the column- and row-weights are $k$ and $r$, respectively.

\vspace{2mm}
\begin{theorem}
An incidence matrix of a BTD$(v,b;\rho_1,\rho_2,r;k,\lambda)$ with $\rho_2=1$ and $\lambda=2$ does not contain $\mathcal{P}_{6}$, $\mathcal{P}_{8}$, and $\mathcal{P}_{10}$.
\end{theorem}
\begin{proof}
Let $P_{\mathrm{BTD}}$ be an incidence matrix of a BTD$(v,b;\rho_1,\rho_2,r;k,\lambda)$ with $\rho_2=1$ and $\lambda=2$. Since every element of this BTD can have multiplicity up to two, $\mathcal{P}_6$ does not appear in $P_{\mathrm{BTD}}$. The condition $\rho_2=1$ implies that 2 appears once in each row of $P_{\mathrm{BTD}}$ and $\lambda=2$ implies that each column of $P_{\mathrm{BTD}}$ can have at most one 2. Hence $P_{\mathrm{BTD}}$ does not contain $\mathcal{P}_8$. Since a pair of distinct elements will appear at least three times in this BTD if $P_{\mathrm{BTD}}$ has $\mathcal{P}_{10}$ as its submatrix, $P_{\mathrm{BTD}}$ does not contain $\mathcal{P}_{10}$.
\end{proof}

\vspace{2mm}
All possible BTDs with $r \leq 15$ are given in \cite{Billington}. Table \ref{table1} lists all parameters of regular protographs with $d_c \leq 15$ avoiding inevitable cycles of length less than 12 constructed from BTDs.

\begin{table}[!t]
\renewcommand{\arraystretch}{1.3}
\caption{Regular Protographs With $n_2=J$ Avoiding Inevitable Cycles of Length $<$ 12 Constructed From BTDs for $d_c \leq 15$}
\label{table1}
\centering
\begin{tabular}{|c||c|c|c|c|c|c|c|c|c|c|c|}
\hline
$J$ & 6 & 12 & 9 & 20 & 12 & 30 & 42 & 48 & 42 & 15 & 60\\
\hline
$L$ & 12 & 24 & 27 & 40 & 48 & 60 & 63 & 64 & 84 & 75 & 100\\
\hline
$d_v$ & 3 & 4 & 3 & 5 & 3 & 6 & 8 & 9 & 7 & 3 & 9\\
\hline
$d_c$ & 6 & 8 & 9 & 10 & 12 & 12 & 12 & 12 & 14 & 15 & 15\\
\hline
\end{tabular}
\end{table}

\vspace{2mm}
\begin{example}
An incidence matrix of BTD$(6,12;4,1,6;3,2)$ is shown in Fig. \ref{fig:6_12} and we can see that any ICI subgraph $\mathcal{P}_{2i}$ for $i \leq 5$ does not appear.

%\begin{equation*}
%\left[
%	\begin{tabular}{*{12}{c}}
%		2 & 1 & 1 & 1 & 1 & 0 & 0 & 0 & 0 & 0 & 0 & 0 \\
%		1 & 0 & 0 & 0 & 0 & 2 & 1 & 1 & 1 & 0 & 0 & 0 \\
%		0 & 2 & 0 & 0 & 0 & 1 & 0 & 0 & 0 & 1 & 1 & 1 \\
%		0 & 0 & 2 & 0 & 0 & 0 & 1 & 1 & 0 & 1 & 1 & 0 \\
%		0 & 0 & 0 & 2 & 0 & 0 & 1 & 0 & 1 & 1 & 0 & 1 \\
%		0 & 0 & 0 & 0 & 2 & 0 & 0 & 1 & 1 & 0 & 1 & 1
%     \end{tabular}
%\right].
%\end{equation*}
\end{example}

\vspace{2mm}
As in Table \ref{table1}, the incidence matrices of BTDs with $\rho_2=1$ and $\lambda=2$ do not provide a sufficiently large number of regular protographs. In fact, the condition that every pair of distinct elements appears exactly twice is not necessary and the condition that every pair of distinct elements appears at most twice is enough for constructing regular protographs avoiding inevitable cycles of length less than 12. Besides the regular protographs in Table \ref{table1}, there are many regular protographs with $n_2=J$ avoiding inevitable cycles of length less than 12.

\vspace{2mm}
\begin{example}
Find the smallest regular protograph with $d_v=3$ and $n_2=J$ avoiding inevitable cycles of length less than 12. We first derive a necessary condition for the existence of such a regular protograph by regarding the protograph as an incidence matrix of a block design. There are ${J \choose 2}$ distinct pairs of elements and, on the other hand, the number of all possible pairs of elements in the design is $n_2 \cdot 2 + (L-n_2) \cdot {3 \choose 2}$. Since every pair of elements appears at most twice, that is, $\mathcal{P}_{10}$ does not appear in the protograph, we have a necessary condition
\begin{equation}
2 \cdot {J \choose 2} \geq n_2 \cdot 2 + (L-n_2) \cdot {3 \choose 2}.
\label{eq:12_condition}
\end{equation}

For $J=3$, due to $L \geq 4$ and $n_2=3$, the necessary condition (\ref{eq:12_condition}) is not satisfied.
For $J=4$, by counting the edges in the protograph, the equality $Jd_c=d_vL$, that is, $4d_c=3L$ holds. Since the smallest integer root of this equality is $(d_c,L)=(6,8)$, we have $L \geq 8$ and (\ref{eq:12_condition}) is not satisfied. Similarly, for $J=5$, $L$ should be larger than or equal to 10 and (\ref{eq:12_condition}) is not satisfied either.

For $J=6$, from $6d_c=3L$, the possible smallest protograph has the size $6 \times 8$ and it satisfies (\ref{eq:12_condition}). By first constructing a $6 \times 6$ regular matrix where each column has one 2 and then properly adding two columns only consisting of 0's and 1's, a $6 \times 8$ regular protograph can be constructed as given in Fig. \ref{fig:6_8}. This is the smallest regular protograph with $d_v=3$ and $n_2=J$ avoiding inevitable cycles of length less than 12.
\end{example}

\begin{figure}[!t]
	\centering
	\subfigure[$6 \times 12$]{\includegraphics[scale=0.8]{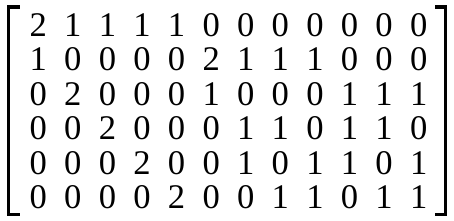}%
	\label{fig:6_12}} \hspace{5mm}
	\subfigure[$6 \times 8$]{\includegraphics[scale=0.8]{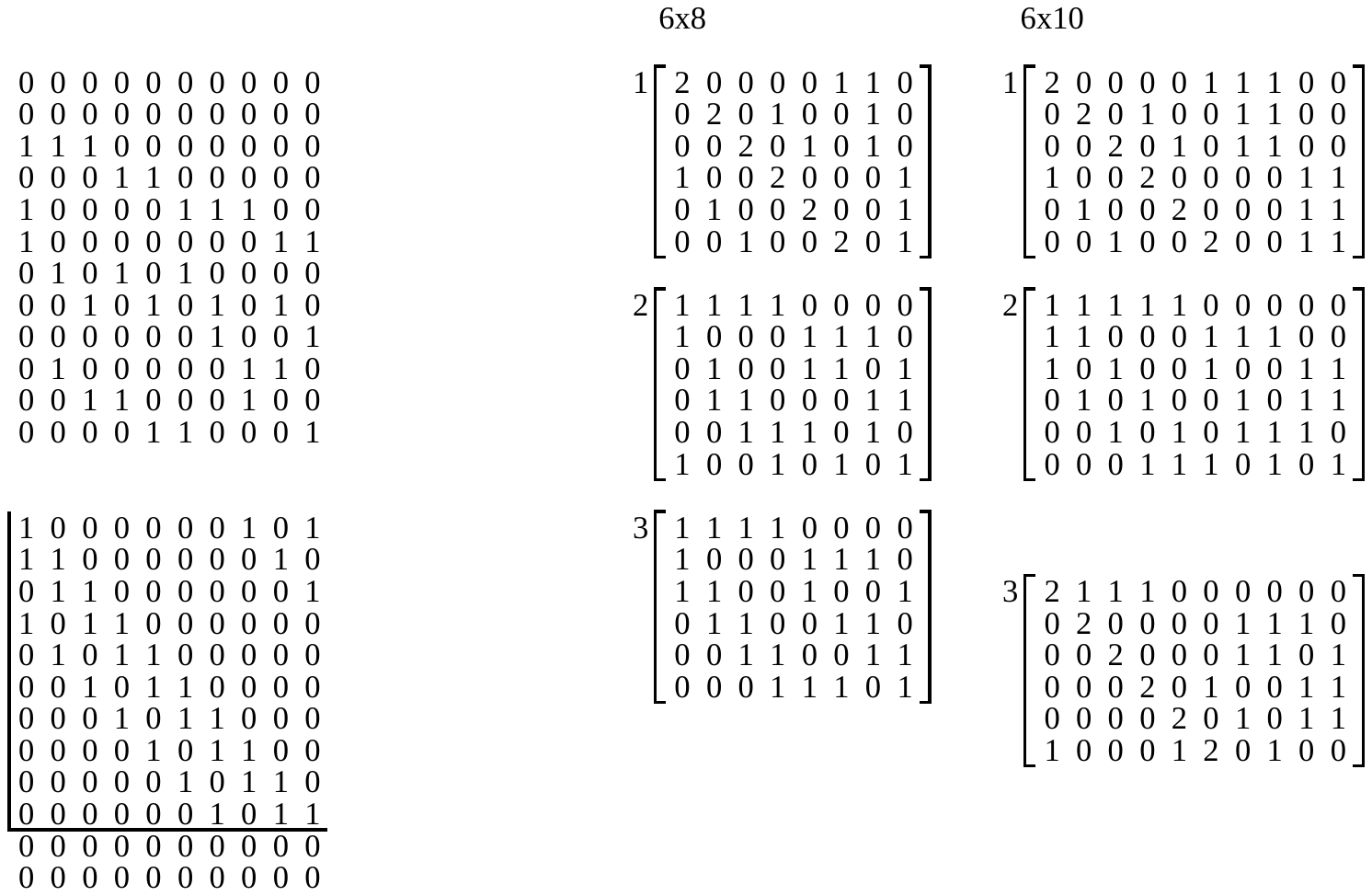}%
	\label{fig:6_8}}
	\caption{Two regular protographs with $d_v=3$ and $n_2=6$ avoiding inevitable cycles of length $<$ 12.}
	\label{fig:12}
\end{figure}

\vspace{2mm}
\section{Construction of Regular Protographs Avoiding Inevitable Cycles of Length Less Than 14}\label{sec:14}

Now we will focus on the construction of regular multiple-edge protographs avoiding inevitable cycles of length less than 14. A systematic construction method of single-edge regular protographs avoiding inevitable cycles of length less than 14 was provided in \cite{Kim}. Since multiple-edge protographs are now being considered, two additional ICI subgraphs having double edges of $\mathcal{P}_{12}$ as well as $\mathcal{P}_8$ and $\mathcal{P}_{10}$ must be avoided, which makes the problem more complicated. In this section, systematic construction methods for multiple-edge protographs are proposed based on various combinatorial designs.

Consider a regular $J \times L$ protograph whose column- and row-weights are $d_v$ and $d_c$, respectively. Let $n_2$ denote the number of double edges in the protograph. Assume that $d_v \geq 3$ because the regular QC LDPC codes with $d_v=2$ are not used in general due to their poor performance.

Using row and column permutations, every regular protograph not inducing inevitable cycles of length less than 14 can be represented as in Fig. \ref{fig:structure}. The $n_2 \times n_2$ submatrix $A$ has $n_2$ 2's as its diagonal elements and the other elements of $A$ should be zero to avoid the first ICI subgraph of $\mathcal{P}_{12}$. $F$ is a $(J-n_2) \times n_2$ submatrix consisting of columns of weight $d_v-2$. By appropriate column permutation of all but $A$ and $F$ in the protograph, all the columns whose lower parts have nonzero weight are relocated in the part of $B$ and $G$, and the remaining columns make $T$ with column-weight $d_v$ and all-zero matrix $O$. Let $G$ and $T$ be $J_G \times L_G$ and $J_T \times L_T$ matrices, respectively.

By Lemma \ref{lemma:1}, $n_2$ cannot be larger than $J$. Moreover, if the regular protographs which do not induce inevitable cycles of length less than 14 are considered for $d_v=3$, the following theorem provides an additional condition on $n_2$.

\begin{figure}[tb]
    \centering
	\includegraphics[scale=0.7]{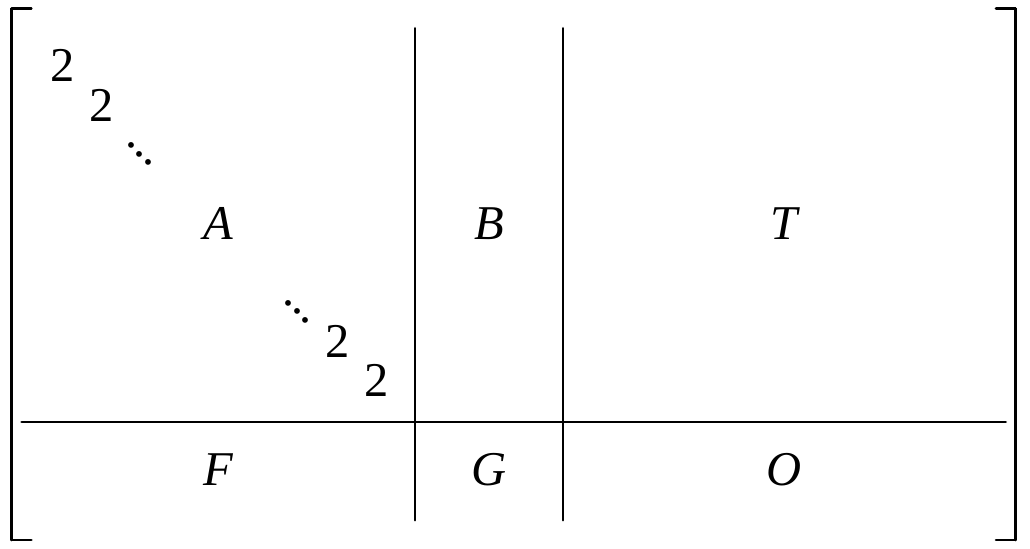}%
    \caption{The structure of regular protographs avoiding inevitable cycle of length $<$ 14.}
    \label{fig:structure}
\end{figure}

\vspace{2mm}
\begin{theorem}
Assume that a regular protograph with $d_v = 3$ and $d_c \geq 4$ does not induce inevitable cycles of length less than 14. Then $n_2 \leq J-2$.
\label{thm:ntwo}
\end{theorem}
\begin{proof}
The inequality $n_2 \leq J$ holds by Lemma \ref{lemma:1}. The protograph with $n_2=J$ should be of the form $[A|B|T]$ from Fig. \ref{fig:structure} and the submatrix $A$ is no longer a diagonal matrix due to $d_v=3$. Therefore, $A$ should contain the first ICI subgraph of $\mathcal{P}_{12}$ because each column of $A$ also contains exactly one 1 and hence $n_2$ should be less than $J$. There should be one double edge in each row to avoid $\mathcal{P}_8$ and each column should have at most one double edge. The column with a double edge also has a single edge in other position, which generates the first pattern of $\mathcal{P}_{12}$.

Now suppose that $n_2=J-1$. The protograph has the form of Fig. \ref{fig:structure} and $F$ is the $1 \times (J-1)$ all-1 matrix. Due to $F$, $d_c$ cannot be less than $J-1$. If $d_c > J-1$, $G$ becomes the $1 \times (d_c-(J-1))$ all-1 matrix and each column of $B$ has a pair of 1's, which generates $\mathcal{P}_{10}$ in the union of $A$, $B$, $F$, and $G$. If $d_c=J-1$, the protograph is made up of only $A$, $T$, $F$, and $O$, and the size of $T$ is $(J-1) \times (J-1)(J-3)/3$ because the column- and row-weights of $T$ are $3$ and $J-3$, respectively. Since $T$ should not have $\begin{bmatrix} 1 & 1 \\ 1 & 1 \end{bmatrix}$ as its submatrix to avoid the second and the third ICI subgraphs of $\mathcal{P}_{12}$ in the union of $A$, $B$, and $T$, a pair of 1's in the same column can appear at most once in $T$. To satisfy this condition, the number of all possible column-wise pairs of 1's should be larger than or equal to the number of actual column-wise pairs of 1's in $T$. Therefore, we have ${J-1 \choose 2} \geq {3 \choose 2} \cdot (J-1)(J-3)/3$, i.e., $J \leq 4$. Due to $d_c=J-1$, this contradicts the assumption of $d_c \geq 4$.
\end{proof}

\vspace{2mm}
Based on Theorem \ref{thm:ntwo}, the case of $d_v=3$ and $n_2=J-2$ is considered in Subsection \ref{subsec:J-2} and the construction method of regular protographs for $d_v=3$ and $n_2=J-2$ is extended not only to the case of $d_v=3$ and $n_2 < J-2$ but also to the case of $d_v \geq 4$ in Subsection \ref{subsec:other}.

\subsection{Regular Protographs With $d_v=3$ and $n_2 = J-2$}
\label{subsec:J-2}

In this subsection, the construction of regular protographs with $d_v=3$ and $n_2=J-2$ is elaborated. Necessary conditions on $d_c$ and $J$ for the existence of regular protographs with $d_v=3$ and $n_2 = J-2$, which avoid inevitable cycles of length less than 14, are derived as follows.

\vspace{2mm}
\begin{theorem}
Assume that a regular protograph with $d_v = 3$, $d_c \geq 4$, and $n_2=J-2$ does not induce inevitable cycles of length less than 14. Then $d_c$ and $J$ should satisfy either
	\begin{enumerate}
		\item $J \equiv 5 \mod 6$, $J\geq 11$, and $d_c=(J+1)/2$ or
		\item $J \equiv 2 \mod 6$, $J\geq 14$, and $d_c=(J-2)/2$ or
		\item $J \equiv 3 \mod 6$, $J\geq 9$, and $d_c=(J-1)/2,~(J+1)/2$ or
		\item $J \equiv 1 \mod 6$, $J\geq 13$, and $d_c=(J-1)/2$ or
		\item $J \equiv 0 \mod 6$, $J\geq 12$, and $d_c=(J-2)/2,~J/2$ or
		\item $J=10$ and $d_c=6$.
	\end{enumerate}	
\label{theorem:main}
\end{theorem}
\begin{proof}
By counting the edges in the protograph, we have $d_cJ=d_vL$. Since $d_v=3$ and $L$ is an integer, $d_cJ \equiv 0 \mod 3$. Also, the submatrix $F$ in Fig. \ref{fig:structure} is a $2 \times (J-2)$ matrix consisting of weight-1 columns. Consider two cases: (i) $F$ contains an all-1 row, (ii) $F$ does not contain any all-1 rows.

For the case (i), if $d_c>J-2$, the ICI subgraph $\mathcal{P}_{10}$ appears in the union of $A$, $B$, $F$, and $G$. If $d_c=J-2$, $G$ is a $2 \times (J-2)$ matrix with an all-1 row at the different row position from the all-1 row of $F$. Then there exist some rows containing a pair of 1's in $B$ because $B$ has $2(J-2)$ 1's and the column-weight of $B$ is 2, which generates the second ICI subgraph of $\mathcal{P}_{12}$ in the union of $A$, $B$, and $G$. Therefore, the case (i) is impossible.

For the case (ii), if a column of $G$ has a pair of 1's, the column including this pair in the protograph and another column in the union of $A$ and $F$ generate $\mathcal{P}_{10}$. Therefore, each column of $F$ and $G$ cannot have a pair of 1's. Since the number of columns in $F$ is $J-2$ and the total number of columns in $F$ and $G$ is $2d_c$, we have $(J-2)/2 \leq d_c$, where the equality holds when $B$ and $G$ do not appear in the protograph. If a row of $B$ has a pair of 1's, either $\mathcal{P}_{10}$ or the second ICI subgraph of $\mathcal{P}_{12}$ must occur in the union of $A$, $F$, $G$, and $B$. Therefore, each row of $B$ can have at most one 1 so that the number of 1's in $B$ cannot exceed the number of rows in $B$. Since the column-weight of $B$ is 2 and $B$ has $2(2d_c-(J-2))$ 1's, we obtain $d_c \leq 3(J-2)/4$ from $2(2d_c - (J-2)) \leq J-2$. Finally, it remains to determine the structure of $T$ such that the submatrix $\begin{bmatrix} 1 & 1 \\ 1 & 1 \end{bmatrix}$ does not appear in the union of $B$ and $T$ to prevent the second ICI subgraph of $\mathcal{P}_{12}$. As in the proof of Theorem \ref{thm:ntwo}, by counting the number of column-wise pairs of 1's in $B$ and $T$, we obtain the condition ${2 \choose 2} \cdot (2d_c - (J-2)) + {3 \choose 2} \cdot (d_cJ/3-2d_c) \leq {J-2 \choose 2}$ yielding $d_c \leq (J-1)(J-2)/(2(J-4))$.

The above conditions on $d_c$ and $J$ are summarized as follows:
	\begin{align*}
		&~~~~~~~~~~~~~~d_c J \equiv 0 \mod 3;\\
		&\frac{J-2}{2} \leq d_c \leq \min \left \{ \frac{3}{4}(J-2),\frac{(J-1)(J-2)}{2(J-4)} \right \}.
%		&\begin{cases}
%			\frac{J-2}{2} \leq d_c \leq \frac{3}{4}(J-2),~~~~J =9, 10 \\
%			\frac{J-2}{2} \leq d_c \leq \frac{(J-1)(J-2)}{2(J-4)},~~J \geq 11.
%		\end{cases}
	\end{align*}
Since $d_c$ and $J$ are integers, the above conditions reduce to simple linear relations with respect to $J$ modulo 6 as given in the theorem statement.
\end{proof}

\vspace{2mm}
In Theorem \ref{theorem:main}, all possible regular protographs avoiding inevitable cycles of length less than 14 are provided for $d_v = 3$ and $n_2=J-2$, and Table \ref{table2} only lists those for $J \leq 26$ among them.

	\begin{table}
		\renewcommand{\arraystretch}{1.3}
		\caption{All Possible Regular Protographs Avoiding Inevitable Cycles of Length $<$ 14 When $d_v=3$ and $n_2=J-2$ for $J \leq 26$}
		\label{table2}
		\centering
		\begin{tabular}{|c||c|c|c|c|c|c|c|c|c|}
			\hline
			$J$ & 9 & 10 & 11 & 12 & 13 & 14 & 15 & 16 & 17 \\
			\hline
			\hline
			$L$ & 12, 15 & 20 & 22 & 20, 24 & 26 & 28 & 35, 40 & $-$ & 51 \\
			\hline
			$d_c$ & 4, 5 & 6 & 6 & 5, 6 & 6 & 6 & 7, 8 & $-$ & 9 \\
			\hline
			\hline
			$J$ & 18 & 19 & 20 & 21 & 22 & 23 & 24 & 25 & 26 \\
			\hline
			\hline
			$L$  & 48, 54 & 57 & 60 & 70, 77 & $-$ & 92 & 88, 96 & 100 & 104 \\
			\hline
			$d_c$  & 8, 9 & 9 & 9 & 10, 11 & $-$ & 12 & 11, 12 & 12 & 12 \\
			\hline
		\end{tabular}
	\end{table}

Now we focus on the existence problem and the construction of the regular protographs with the parameters found in Theorem \ref{theorem:main}. Note that the proposed protographs we will construct may not be all instances with the parameters in Theorem \ref{theorem:main} but we provide at least one instance per each set of parameters and also note that $J_G=2$, $L_G=2d_c-(J-2)$, $J_T=J-2$, and $L_T = d_c(J-6)/3$. For given $J$ and $d_c$, the matrices $B$, $T$, $F$, and $G$ can be constructed step by step as follows:
\begin{enumerate}
	\item[1.] For constructing $B$ and $T$ at once, an incidence matrix of a combinatorial block design suitably chosen for each case in Theorem \ref{theorem:main} is modified such that it has the size $J_T \times (L_G + L_T)$, each of $L_G$ columns corresponding to $B$ has a disjoint pair of 1's, the other columns have the weight 3, all rows have the weight $d_c -2$, and any column-wise pair of 1's appears at most once to avoid the second and the third ICI subgraphs of $\mathcal{P}_{12}$ in $[A|B|T]$.
	\item[2.] In $G$, 1's are placed such that $\lfloor L_G / 2 \rfloor$ columns have 1's in the first row and the other columns have 1's in the second row.
	\item[3.] For constructing $F$, 1's are placed such that the union of $A$, $B$, $F$, and $G$ does not contain $\mathcal{P}_{10}$.
\end{enumerate}
Note that the placement of 1's in the third step is guaranteed by the bound $d_c \leq 3(J-2)/4$ in the proof of Theorem \ref{theorem:main}.

Since the conditions in Theorem \ref{theorem:main} are necessary ones for the existence of $T$, a protograph may not exist for some parameter values. Therefore, we will show that there exist protographs for all parameter values given in Theorem \ref{theorem:main} by providing explicit construction methods of $B$ and $T$ using various combinatorial designs as follows.

\vspace{2mm}
\indent \textit{1)} $J \equiv 5 \mod 6$ and $J \geq 11$:

\vspace{2mm}
In this case, we have $d_c = (J+1)/2$, $L_G = 3$, and $L_T=(J+1)(J-6)/6$. We need to construct $[B|T]$ of size $(J-2) \times (J^2 -5J+12)/6$ to avoid repeated column-wise pairs of 1's, i.e., the subgraph $\begin{bmatrix} 1 & 1 \\ 1 & 1 \end{bmatrix}$. For this, the following Steiner system can be used.

\vspace{2mm}
\begin{definition}[\cite{Colbourn1}]
A $t$-$(v,k,\lambda)$ design is a pair $(V,B)$, where $V$ is a $v$-set of points and $B$ is a collection of $k$-subsets (blocks) of $V$ with the property that every $t$-subset of $V$ is contained in exactly $\lambda$ blocks in $B$. A \textit{Steiner system} $S(t,k,v)$ is the $t$-$(v,k,\lambda)$ design with $\lambda=1$.
\end{definition}

\vspace{2mm}
\begin{lemma}[\cite{Colbourn1}]
There exists $S(2,3,v)$ only when $v \equiv 1,3 \mod 6$.
\label{lemma:existence}
\end{lemma}

\vspace{2mm}
The number of blocks in $S(2,3,v)$ is $v(v-1)/6$. Since three columns have the weight two and the other columns have the weight three in the $(J-2) \times (J^2 -5J+12)/6$ matrix $[B|T]$, the $(J-2) \times (J-2)(J-3)/6$ incidence matrix of $S(2,3,J-2)$ may be modified to be used as $[B|T]$ by deleting one 1 from each of well-chosen three columns and adding one column of weight three. In order for such a modified matrix to be a valid $[B|T]$, we should check whether three column-wise pairs of 1's in the weight-2 columns are disjoint, all rows have the weight $(J-3)/2$, and any column-wise pair of 1's appears at most once.

Without loss of generality, let $\{v_1,v_2,v_i\}$, $\{v_2,v_3,v_j\}$, and $\{v_1,v_3,v_k\}$, $i \neq j \neq k$, be three blocks of $S(2,3,J-2)$ corresponding to three columns containing a cycle of length 6.
Three disjoint blocks $\{v_2,v_i\}$, $\{v_3,v_j\}$, and $\{v_1,v_k\}$ are obtained by removing $v_1$, $v_2$, and $v_3$ from $\{v_1,v_2,v_i\}$, $\{v_2,v_3,v_j\}$, and $\{v_1,v_3,v_k\}$, respectively.
Inserting a block $\{ v_1,v_2,v_3 \}$ to this modified $S(2,3,J-2)$ still makes every pair appear at most once.
An incidence matrix of $S(2,3,J-2)$ has the row-weight $(J-3)/2$ and the above modification clearly keeps the row-weight unchanged. 
Therefore, we propose a construction method of $[B|T]$ in the case of $J \equiv 5 \mod 6$ and $J \geq 11$ as follows:

\begin{enumerate}
	\item[1.] Permute the columns of an incidence matrix of $S(2,3,J-2)$ so that the first three columns contain a cycle of length 6.
	\item[2.] Delete a 1 on the cycle of length 6 from each of the first three columns so that the resulting three column-wise pairs of 1's are disjoint.
	\item[3.] Insert one column of weight three where three 1's are located in the rows passed through by the above cycle of length 6.
%	\item[4.] Perform row permutation of $T$ not to have repeated pairs in the union of $B$ and $T$.
\end{enumerate}

Actually, it is easy to choose three columns which contain a cycle of length 6 because an incidence matrix of $S(2,3,J-2)$ has many cycles of length 6. The following lemma shows how many cycles of length 6 exist in an incidence matrix of $S(2,3,J-2)$.

\vspace{2mm}
\begin{lemma}
An incidence matrix of $S(2,3,J-2)$ has $(J-2)(J-3)(J-5)/6$ cycles of length 6.
\end{lemma}
\begin{proof}
Consider three points $v_1,v_2,v_3 \in V$ of $S(2,3,J-2)$. Three pairs $\{ v_1, v_2 \}, \{ v_2, v_3 \}, \{ v_3, v_1 \}$ appear in $S(2,3,J-2)$ in either of two ways: (i) one block has all the three pairs, that is, consists of $v_1,v_2,v_3$; or (ii) each pair is contained in a block which does not have the other two pairs, that is, there are three blocks $\{ v_1, v_2, v_i \}, \{ v_2, v_3, v_j \}, \{ v_3, v_1, v_k \}$, where $v_i,v_j,v_k \in V$ and $i \neq j \neq k$. Three pairs in the case (ii) form a cycle of length 6 in the incidence matrix of $S(2,3,J-2)$. Hence the number of cycles of length 6 in the incidence matrix can be enumerated by substracting the number of all blocks from the number of the ways of choosing three points in $V$. This yields ${J-2 \choose 3} - (J-2)(J-3)/6 = (J-2)(J-3)(J-5)/6$.
\end{proof}

\begin{figure}[tb]
	\centering
	\subfigure[An incidence matrix of $S(2,3,9)$.]{\includegraphics[scale=0.8]{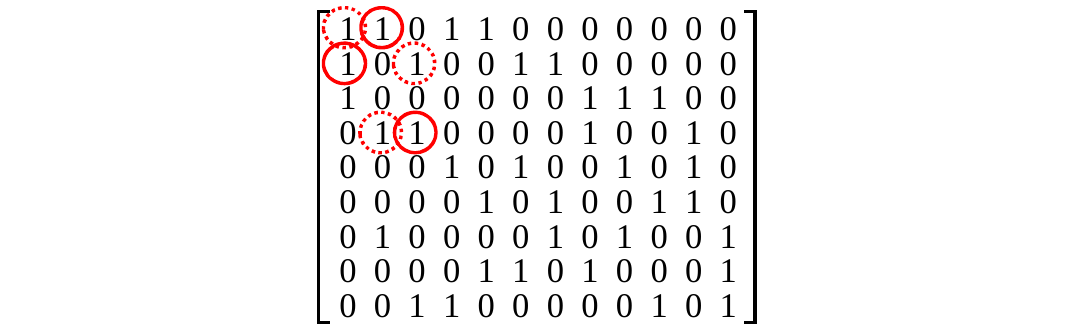}%
	\label{fig:11_22_const}}\\
	\subfigure[An $11 \times 22$ regular protograph with $d_v=3$ and $n_2=9$.]{\includegraphics[scale=0.8]{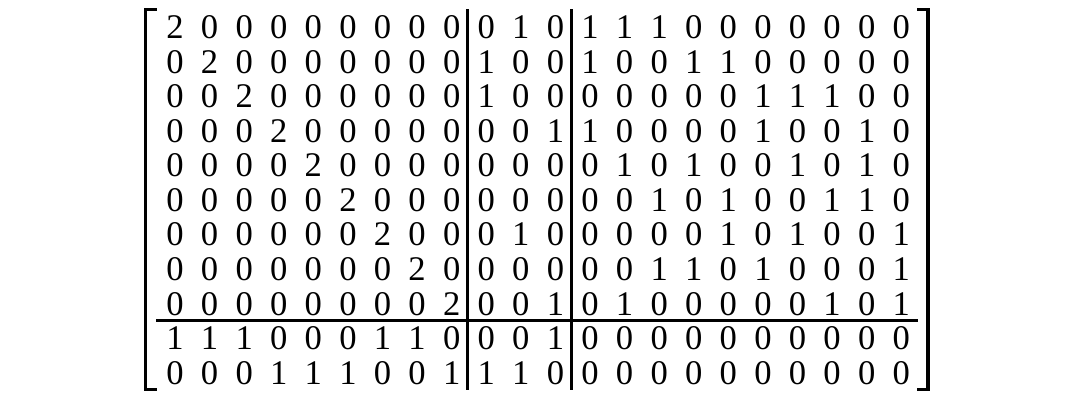}%
	\label{fig:11_22_proto}}
	\caption{The construction of an $11 \times 22$ regular protograph with $d_v=3$ and $n_2=9$.}
	\label{fig:11_22}
\end{figure}

\vspace{2mm}
\begin{example}

Fig. \ref{fig:11_22} illustrates the construction of an $11 \times 22$ protograph with $d_v=3$ and $n_2=9$. A cycle of length 6 is denoted by the circles in the incidence matrix of $S(2,3,9)$ which has been already column-wisely permuted in Fig. \ref{fig:11_22_const}. To obtain $[B|T]$, three 1's marked by dotted circles are deleted and the column with 1's in the first, the second, and the fourth rows is inserted as the first column of $T$. Let $v_i,~i=1,\ldots,9$, denote the points of $S(2,3,9)$, which also denotes the $i$-th row of $[B|T]$. We can see that $[B|T]$ does not have three pairs of 1's $\{v_1,v_3\},\{v_4,v_7\},\{v_2,v_9\}$ in any column and three pairs $\{v_2,v_3\},\{v_1,v_7\},\{v_4,v_9\}$ in $B$ are disjoint. The resulting $11 \times 22$ protograph with $d_v=3$ and $n_2=9$ is shown in Fig. \ref{fig:11_22_proto} and we can check that $\mathcal{P}_{2i}$ with $i \leq 6$ does not appear in this protograph.

\end{example}

\vspace{2mm}
\indent \textit{2)} $J \equiv 2 \mod 6$ and $J \geq 14$:

\vspace{2mm}
In this case, we have $d_c = (J-2)/2$, $L_G = 0$, and $L_T=(J-2)(J-6)/6$. Since $B$ and $G$ do not appear in the protograph, $T$ should be designed to avoid repeated column-wise pairs, where $T$ has constant row-weight $(J-6)/2$ and column-weight 3. A configuration whose incidence matrix has the column-weight 3 and the size $(J-2) \times (J-2)(J-6)/6$ can be used for $T$.

\vspace{2mm}
\begin{definition}[\cite{Colbourn1}]
A \textit{configuration} $(v_r,b_k)$ is an incidence structure of $v$ points and $b$ blocks such that (i) each block contains $k$ points, (ii) each point lies on $r$ blocks, and (iii) two different points are contained in at most one block. If $v=b$ and hence $r=k$, the configuration is called \textit{symmetric} and denoted by $v_k$.
\end{definition}

\vspace{2mm}
It is important to check the existence of the configuration with the required parameters. The following theorem shows that such configuration always exists and therefore $T$ can be constructed.

\vspace{2mm}
\begin{theorem}
There exists a configuration $(v_r,b_k)$ with $v=J-2$, $b= (J-2)(J-6)/6$, $k=3$, and $r=(J-6)/2$ for all $J \equiv 2 \mod 6$ and $J \geq 14$.
\label{thm:existence}
\end{theorem}
\begin{proof}
Necessary conditions for the existence of $(v_r,b_k)$ configuration \cite{Gropp1} are given as (i) $v \leq b$ and $k \leq r$, (ii) $vr=bk$, and (iii) $v \geq r(k-1)+1$. We can easily check that the parameters in the theorem statement satisfy these conditions. Finally, the existence of such configurations is guaranteed by Theorem 3.1 in \cite{Gropp1}, that is, there exists a configuration with $k=3$ if and only if the necessary conditions hold.
\end{proof}

\vspace{2mm}
Now, a construction method of $T$ is proposed based on the results in \cite{Gropp1}, which uses configurations with parallel classes and resolvable configurations.

\vspace{2mm}
\begin{definition}[\cite{Colbourn1}]
A \textit{parallel class} in a design is a set of blocks that partition the point set. A \textit{resolvable design} is a design whose blocks can be partitioned into parallel classes.
\end{definition}

\vspace{2mm}
For $J \equiv 2 \mod 6$, $S(2,3,J-1)$ exists by Lemma \ref{lemma:existence}. For $J \geq 20$, an incidence matrix of a resolvable configuration $(v_r,b_k)$ with $v=J-2$, $b= (J-2)(J-4)/6$, $k=3$, and $r=(J-4)/2$ can be constructed by removing a row and its incident columns in an incidence matrix of $S(2,3,J-1)$ \cite{Gropp1}. For $J=14$, there is no resolvable configuration $(12_5,20_3)$ but we can find a configuration $(12_5,20_3)$ in the same manner as illustrated in Fig. \ref{fig:14_28_const2}, which contains some parallel classes from $S(2,3,13)$ \cite{Gropp1}. Since a parallel class of a configuration $(v_r,b_k)$ with $v=J-2$, $b= (J-2)(J-4)/6$, $k=3$, and $r=(J-4)/2$ consists of $(J-2)/3$ blocks and has all points exactly once, we obtain $T$ by removing one parallel class from the incidence matrices of these configurations. The construction procedure of $T$ for $J \equiv 2 \mod 6$ and $J \geq 14$ is summarized as:
			
	\begin{enumerate}
		\item[1.] Construct $S(2,3,J-1)$.
		\item[2.] Make an incidence matrix of a resolvable configuration $(v_r,b_k)$ with $v=J-2$, $b= (J-2)(J-4)/6$, $k=3$, and $r=(J-4)/2$ by removing a row and its incident columns in an incidence matrix of $S(2,3,J-1)$.
		\item[3.] Remove one parallel class which consists of $(J-2)/3$ columns to obtain $T$.
	\end{enumerate}

\begin{figure}[tb]
	\centering
	\subfigure[An incidence matrix of $S(2,3,13)$.]{\includegraphics[scale=0.8]{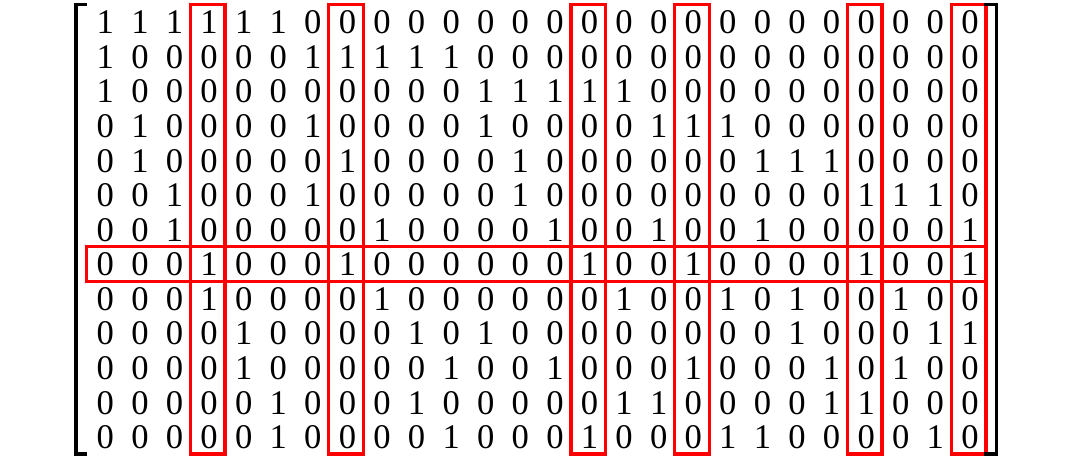}%
	\label{fig:14_28_const1}}\\
	\subfigure[An incidence matrix of a configuration $(12_5,20_3)$.]{\includegraphics[scale=0.8]{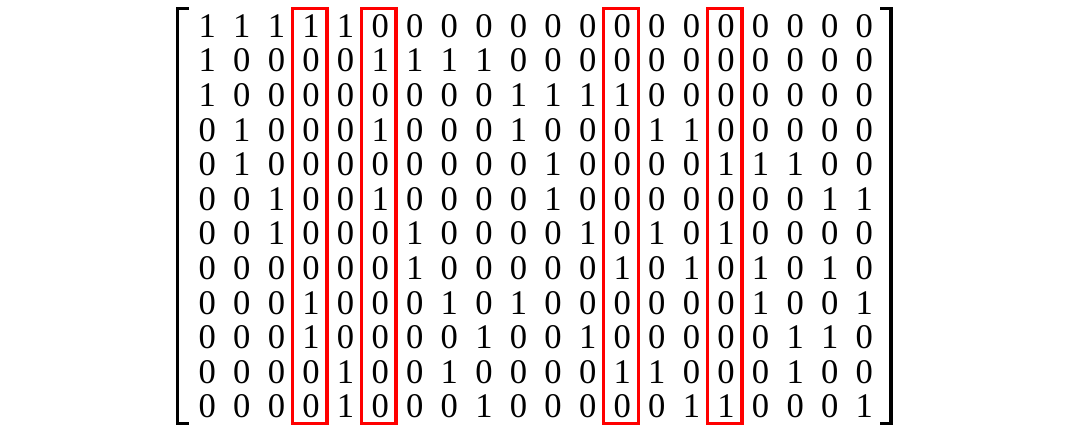}%
	\label{fig:14_28_const2}}\\
	\subfigure[A $14 \times 28$ regular protograph with $d_v=3$ and $n_2=12$.]{\includegraphics[scale=0.8]{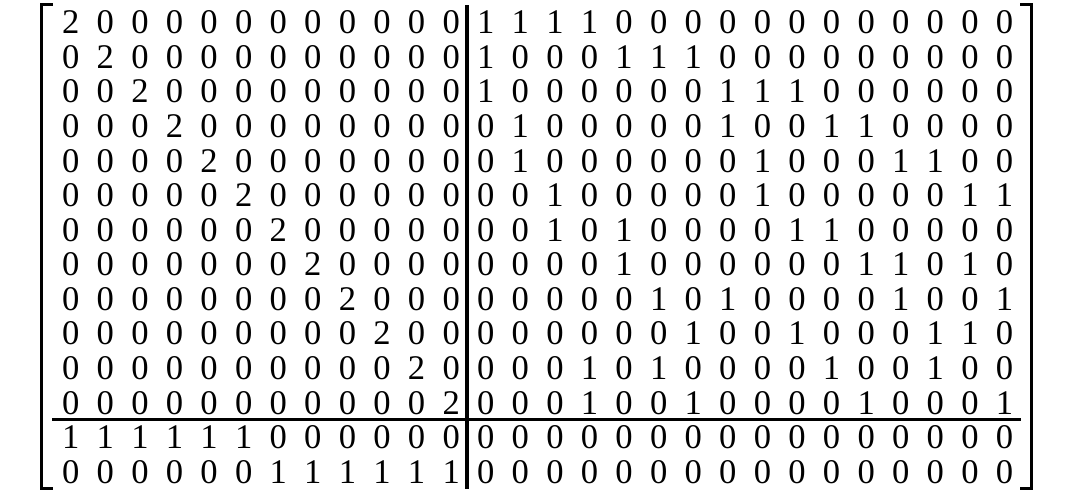}%
	\label{fig:14_28_proto}}
	\caption{The construction of a $14 \times 28$ regular protograph with $d_v=3$ and $n_2=12$.}
	\label{fig:14_28}
\end{figure}

\vspace{2mm}
\begin{example}
An incidence matrix of $S(2,3,13)$ is shown in Fig. \ref{fig:14_28_const1}. An incidence matrix of a configuration $(12_5,20_3)$ in Fig. \ref{fig:14_28_const2} is constructed by removing the eighth row and its incident columns in the incidence matrix of $S(2,3,13)$ in Fig. \ref{fig:14_28_const1}. We see that the fourth, the sixth, the thirteenth, and the sixteenth columns form a parallel class. By removing these columns, an incidence matrix of a configuration $(12_4,16_3)$ is constructed, which is used as $T$. The resulting $14 \times 28$ protograph with $d_v=3$ and $n_2=12$ is shown in Fig. \ref{fig:14_28_proto}.
\end{example}

\vspace{2mm}				
\indent \textit{3)} $J \equiv 3 \mod 6$ and $J \geq 9$:

\vspace{1mm}
\quad \textit{3.1)} $d_c = (J-1)/2$ except for $J = 9$;

\vspace{2mm}		
In this case, we have $L_G = 1$ and $L_T=(J-1)(J-6)/6$, and thus $B$ should have only one pair of 1's. Since $S(2,3,J-2)$ exists by Lemma \ref{lemma:existence}, $[B|T]$ may be constructed by removing $J/3-1$ columns from a $(J-2) \times (J-2)(J-3)/6$ incidence matrix of $S(2,3,J-2)$ and then deleting a 1 in other column. To achieve the desired row-weight $(J-5)/2$ of $[B|T]$, there should exist a submatrix consisting of $J/3$ columns satisfying that (i) two rows have weight 2 and the others have weight 1 and (ii) one column has 1 at each of these two rows of weight 2.

As seen in the case of $J \equiv 2 \mod 6$ and $J \geq 14$, for $J \equiv 3 \mod 6$ and $J \geq 15$, $S(2,3,J-2)$ contains as a substructure a configuration $(v_r,b_k)$ with $v=J-3$, $b= (J-3)(J-5)/6$, $k=3$, and $r=(J-5)/2$ which has at least one parallel class consisting of $J/3-1$ blocks. This implies that there are $J/3-1$ blocks in $S(2,3,J-2)$, which partition all but one point. Also, there always exists another block containing that point in $S(2,3,J-2)$. These $J/3$ blocks satisfy the above requirements (i) and (ii) for $[B|T]$. The construction procedure of $[B|T]$ for $J \equiv 3 \mod 6$, $J \geq 15$, and $d_c = (J-1)/2$ is summarized as:

\begin{figure}[tb]
	\centering
	\subfigure[An incidence matrix of $S(2,3,13)$.]{\includegraphics[scale=0.8]{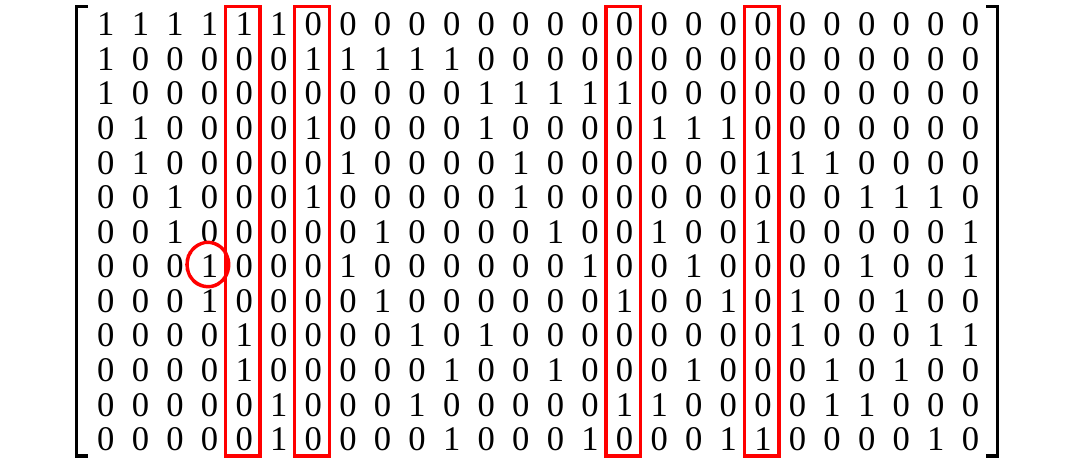}%
	\label{fig:15_35_const}}\\
	\subfigure[A $15 \times 35$ regular protograph with $d_v=3$ and $n_2=13$.]{\includegraphics[width=3.5in]{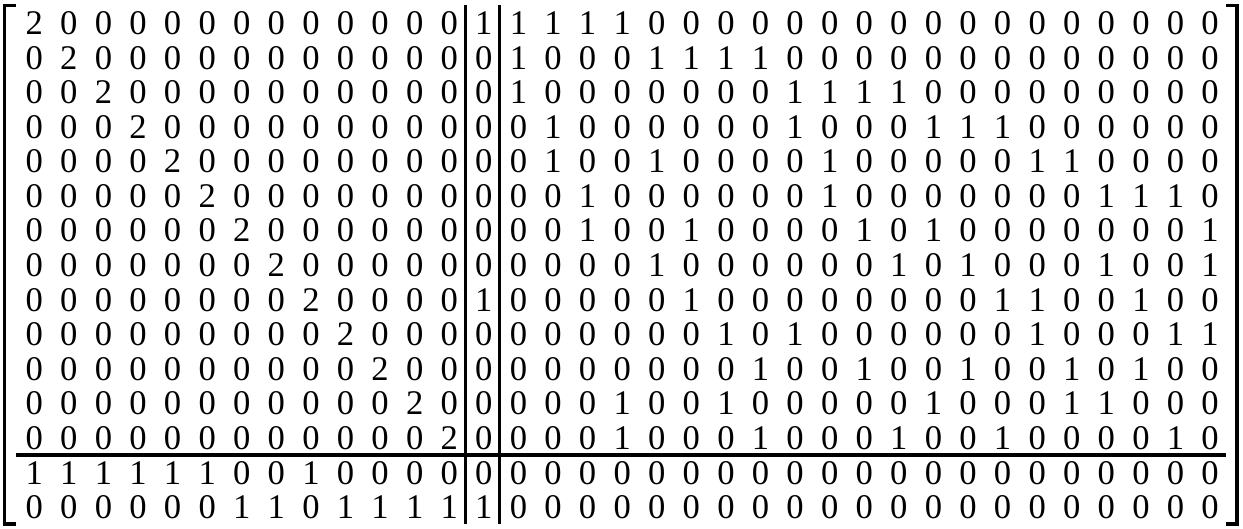}%
	\label{fig:15_35_proto}}
	\caption{The construction of a $15 \times 35$ regular protograph with $d_v=3$ and $n_2=13$.}
	\label{fig:15_35}
\end{figure}			
									
	\begin{enumerate}
		\item[1.] Construct $S(2,3,J-2)$.
		\item[2.] Select one row in an incidence matrix of $S(2,3,J-2)$ such that if the row and its incident columns are removed from the incidence matrix, the remaining part forms an incidence matrix of a configuration $(v_r,b_k)$ with $v=J-3$, $b= (J-3)(J-5)/6$, $k=3$, and $r=(J-5)/2$ including at least one parallel class.
		\item[3.] Find $J/3-1$ columns which form one parallel class in the above configuration.
		\item[4.] In the incidence matrix of $S(2,3,J-2)$, delete a 1 in the selected row in Step 2 and remove the $J/3-1$ columns obtained in Step 3.
		\item[5.] Move the column which had the deleted 1 in Step 4 to the leftmost to obtain $[B|T]$.
	\end{enumerate}

Note that the above construction method cannot be applied to the case of $J=9$ and $d_c=4$ because the configuration $(6_2,4_3)$ does not have any parallel class. The case of $J=9$ and $d_c=4$ will be covered in the last part of this subsection.

\vspace{2mm}
\begin{example}

Fig. \ref{fig:15_35} illustrates the construction of a $15 \times 35$ regular protograph with $d_v=3$ and $n_2=13$. In an incidence matrix of $S(2,3,13)$ in Fig. \ref{fig:15_35_const}, the fifth, the seventh, the sixteenth, and the twentieth columns partition the set of row indices except for the index of the eighth row and the fourth column has a 1 in the eighth row. Thus, the 1 in the fourth column and the eighth row is deleted and the four boxed columns are removed from the incidence matrix. Then the resulting column of weight 2 is moved to the leftmost and a $15 \times 35$ protograph with $d_v=3$ and $n_2=13$ is shown in Fig. \ref{fig:15_35_proto}. 

\end{example}

\vspace{2mm}
\quad \textit{3.2)} $d_c = (J+1)/2$;

\vspace{2mm}
In this case, we have $L_G = 3$ and $L_T=(J+1)(J-6)/6$, and there are three column-wise pairs of 1's in $B$. Since $S(2,3,J-2)$ exists for $J \equiv 3 \mod 6$ and $J \geq 9$ by Lemma \ref{lemma:existence}, the construction method for $J \equiv 5 \mod 6$ and $J \geq 11$ can also be applied to this case in the same way. As an example, a $9 \times 15$ regular protograph with $d_v=3$ and $n_2=7$ is shown in Fig. \ref{fig:9_15}.

\begin{figure}[tb]
	\centering
	\includegraphics[scale=0.8]{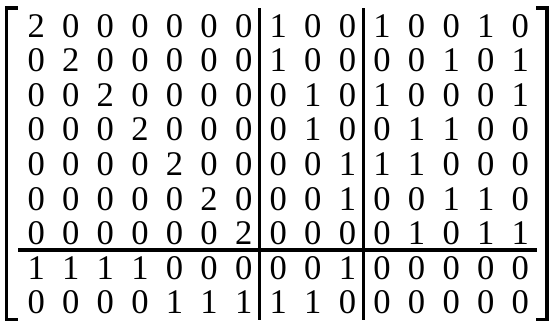}
	\caption{A $9 \times 15$ regular protograph with $d_v=3$ and $n_2=7$.}
	\label{fig:9_15}
\end{figure}

%\begin{example}
%$(J,L) = (15,40),~(d_v,d_c)=(3,8)$
%
%\begin{figure}[tb]
%    \centering
%	\includegraphics[width=3.5in]{15_40_proto}%
%    \caption{$15 \times 40$ protograph with $d_v=3$, $n_2=13$.}
%    \label{fig:15_40}
%\end{figure}
%			
%\end{example}

\vspace{2mm}				
\indent \textit{4)} $J \equiv 1 \mod 6$ and $J \geq 13$:

\vspace{2mm}
In this case, we have $d_c = (J-1)/2$, $L_G = 1$, and $L_T=(J-1)(J-6)/6$.
To construct $[B|T]$, start with $S(2,3,J)$ which always exists by Lemma \ref{lemma:existence}. Similar to the case of $J \equiv 2 \mod 6$ and $J \geq 14$, a configuration $(v_r,b_k)$ with $v=J-1$, $b=(J-1)(J-3)/6$, $k=3$, and $r=(J-3)/2$ can be constructed by removing a row and its incident columns in an incidence matrix of $S(2,3,J)$. Since any $(J-3)/2$ blocks sharing a common point partition all points except the common point and another point in the configuration, by removing any row and its incident columns in an incidence matrix of the configuration, a $(J-2) \times (J-3)(J-4)/6$ matrix with $J-3$ rows of weight $(J-5)/2$ and a row of weight $(J-3)/2$ is obtained. Removing a 1 in the row of weight $(J-3)/2$ results in a matrix which has the desired row-weight $(J-5)/2$ and exactly one column of weight 2.  Clearly, this matrix can be used as $[B|T]$. The construction procedure of $[B|T]$ for $J \equiv 1 \mod 6$ and $J \geq 13$ is summarized as:

	\begin{enumerate}
		\item[1.] Construct a configuration $(v_r,b_k)$ with $v=J-1$, $b=(J-1)(J-3)/6$, $k=3$, and $r=(J-3)/2$ from $S(2,3,J)$ similar to the case of $J \equiv 2 \mod 6$ and $J \geq 14$.
		\item[2.] Obtain a matrix of size $(J-2) \times (J-3)(J-4)/6$ by removing a row and its incident columns.
		\item[3.] Delete a 1 in the row of weight $(J-3)/2$ and move the column having the deleted 1 to the leftmost to obtain $[B|T]$.
	\end{enumerate}

\begin{figure}[tb]
	\centering
	\subfigure[An $11 \times 15$ modified matrix.]{\includegraphics[scale=0.8]{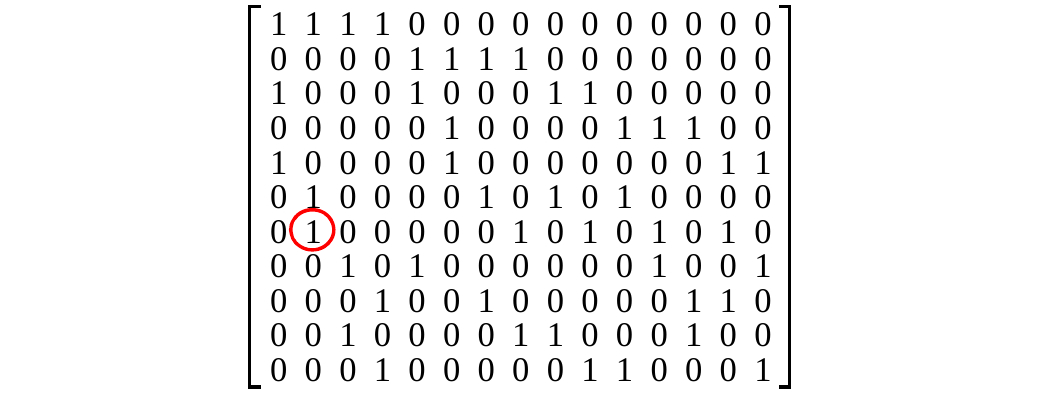}%
	\label{fig:13_26_const1}}\\
	\subfigure[A $13 \times 26$ regular protograph with $d_v=3$ and $n_2=11$.]{\includegraphics[scale=0.8]{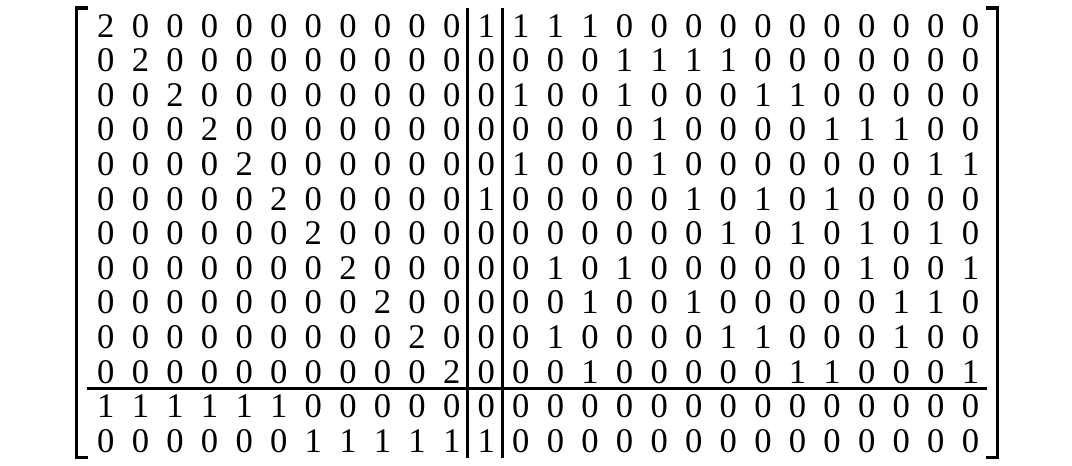}%
	\label{fig:13_26_proto}}
	\caption{The construction of a $13 \times 26$ regular protograph with $d_v=3$ and $n_2=11$.}
	\label{fig:13_26}
\end{figure}

\vspace{1mm}			
\begin{example}
An incidence matrix of a configuration $(12_5,20_3)$ is given in Fig. \ref{fig:14_28_const2}. By removing the first row and its incident columns, we obtain an $11 \times 15$ matrix shown in Fig. \ref{fig:13_26_const1}, where the seventh row has the weight 5 and the others have the weight 4. Then the 1 in the seventh row is deleted from the second column and the second column is moved to the leftmost. The resulting $13 \times 26$ protograph with $d_v=3$ and $n_2=11$ is shown in Fig. \ref{fig:13_26_proto}.
\end{example}

\vspace{2mm}		
\indent \textit{5)} $J \equiv 0 \mod 6$ and $J \geq 12$:

\vspace{1mm}
\quad \textit{5.1)} $d_c=(J-2)/2$;

\vspace{2mm}
In this case, we have $L_G = 0$ and $L_T=(J-2)(J-6)/6$. Similar to the case of $J \equiv 2 \mod 6$ and $J \geq 14$, an incidence matrix of a configuration $(v_r,b_k)$ with $v=J-2$, $b=(J-2)(J-6)/6$, $k=3$, and $r=(J-6)/2$ can be used as $T$. Such configuration can be constructed by using difference triangle set (DTS).

\vspace{2mm}
\begin{definition}[\cite{Colbourn1}]
An $(n,m)$-\textit{difference triangle set}, or $(n,m)$-DTS, is a set $\mathcal{U} = \{ U_1, \ldots, U_n \}$, where for $1 \leq i \leq n$, $U_i= \{ a_{i0}, a_{i1}, \ldots, a_{im} \}$ with $a_{ij}$ an integer satisfying $0=a_{i0} < a_{i1} < \cdots < a_{im}$, and the differences $a_{il} - a_{ij}$ over the integers for all $i$, $j$, $l$,  $1 \leq i \leq n$, $0 \leq l \neq j \leq m$, are all distinct and nonzero.
\end{definition}

\vspace{2mm}
\begin{theorem}[\cite{Gropp1}]
If there is an $(n,2)$-DTS, a configuration $(v_r,b_k)$ for $v\geq 6n+3$, $b=nv$, $k=3$, and $r=nk$ can be constructed from this DTS.
\label{theorem:DTS}
%For $J \equiv 0 \mod 6$, a configuration $(v_r,b_k)$ with $v=J-2$, $b=(J-2)(J-6)/6$, $k=3$, and $r=(J-6)/2$ can be constructed by using $(n,m)$-DTS with $n=(J-6)/6$ and $m=2$.
\end{theorem}
%\begin{proof}
%The above parameters satisfy the conditions on existence and constructibility described in \cite{Gropp1}.
%\end{proof}

\vspace{2mm}
For $J \equiv 0 \mod 6$ and $J \geq 12$, a configuration $(v_r,b_k)$ with $v=J-2$, $b=(J-2)(J-6)/6$, $k=3$, and $r=(J-6)/2$ can be constructed from $((J-6)/6,2)$-DTS by Theorem \ref{theorem:DTS}. According to \cite{Gropp1}, the construction procedure of $T$ is provided as:

\begin{enumerate}
	\item[1.] Construct a $((J-6)/6,2)$-DTS with $U_i = \{ a_{i0}, a_{i1}, a_{i2}\}$, $i=1, \ldots,(J-6)/6$.
	\item[2.] For each $U_i$, $i=1, \ldots, (J-6)/6$, construct a column of length $J-2$ denoted by $C_i$, which has 1 at the $(a_{i0}+1)$-st, the $(a_{i1}+1)$-st, and the $(a_{i2}+1)$-st rows and 0 at other rows.
	\item[3.] For each $i$, construct a $(J-2) \times (J-2)$ matrix whose $j$-th column, $j=1,\ldots,J-2$, is obtained by cyclically shifting $C_i$ downward $j-1$ times.
	\item[4.] Concatenate $(J-6)/6$ matrices in Step 3 to obtain $T$.
\end{enumerate}

Note that we can easily construct a $((J-6)/6,2)$-DTS from the lists of DTS in \cite{web} for $J \equiv 0 \mod 6$ and $J \geq 12$. Fig. \ref{fig:12_20} shows a $12 \times 20$ regular protograph with $d_v=3$ and $n_2=10$ constructed from $(1,2)$-DTS. 

\begin{figure}[tb]
    \centering
	\includegraphics[scale=0.8]{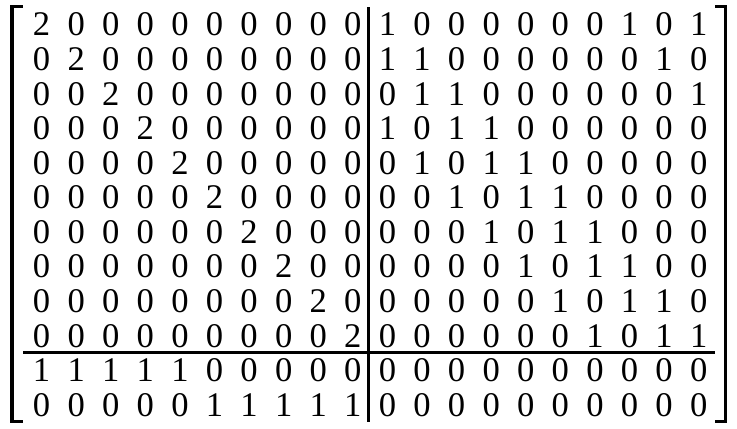}%
    \caption{A $12 \times 20$ regular protograph with $d_v=3$ and $n_2=10$.}
    \label{fig:12_20}
\end{figure}

\vspace{2mm}
\quad \textit{5.2)} $d_c=J/2$;

\vspace{2mm}						
In this case, we have $L_G = 2$ and $L_T=J(J-6)/6$. Pairwise balanced designs (PBDs) can be used to construct $[B|T]$.

\vspace{2mm}				
\begin{definition}[\cite{Colbourn1}]
Let $K$ be a subset of positive integers and let $\lambda$ be a positive integer. A \textit{pairwise balanced design} of order $v$ with block sizes from $K$, denoted by PBD$(v,K;\lambda)$, is a pair $(\mathcal{V},\mathcal{B})$, where $\mathcal{V}$ is a point set of cardinality $v$ and $\mathcal{B}$ is a family of blocks of $\mathcal{V}$ which satisfy that (i) if $B \in \mathcal{B}$, then $|B| \in K$ and (ii) every pair of distinct elements of $\mathcal{V}$ occurs in exactly $\lambda$ blocks of $\mathcal{B}$.
\end{definition}

\vspace{2mm}
Let PBD$(v,K)$ denote a PBD$(v,K;\lambda)$ with $\lambda=1$ and use PBD$(v,K \cup \{ k^\star \})$ to denote a PBD containing only one block of size $k$ in the PBD, where $k \notin K$ is a positive integer. For $J \equiv 0 \mod 6$, it was shown in \cite{Kucukcifci} that PBD$(J-1,\{ 3,5^\star \})$ always exists. Note that five rows sharing 1 with the column of weight 5 have the weight $(J-4)/2$ and the other rows have the weight $(J-2)/2$ in a $(J-1) \times (J^2-3J-12)/6$ incidence matrix of PBD$(J-1,\{ 3,5^\star \})$.

\vspace{2mm}
\begin{theorem}
Removing a row of weight $(J-4)/2$ and its incident columns except the weight-5 column from an incidence matrix of PBD$(J-1,\{ 3,5^\star \})$ makes a $(J-2) \times (J^2 -6J +6)/6$ matrix of constant row-weight $(J-4)/2$.
\label{theorem:PBD_matrix}
\end{theorem}
\begin{proof}
Without loss of generality, assume that the first column has 1's at the first five rows in an incidence matrix of PBD$(J-1,\{ 3,5^\star \})$. Consider the $(J-1) \times (J-4)/2$ submatrix which consists of the columns incident to the first row. Except the first row, each row of this submatrix has only one 1 because every column-wise pair of 1's should appear exactly once in an incidence matrix of this PBD. Thus, the $(J-2) \times (J-6)/2$ submatrix obtained by removing the first row and the first column from the $(J-1) \times (J-4)/2$ submatrix does not have 1 in the first four rows and each of the other rows has only one 1. After removing the first row and the $(J-2) \times (J-6)/2$ submatrix from the incidence matrix of PBD$(J-1,\{ 3,5^\star \})$, the remainder forms the $(J-2) \times (J^2 -6J +6)/6$ matrix of row-weight $(J-4)/2$.
\end{proof}

\vspace{2mm}
The matrix constructed in Theorem \ref{theorem:PBD_matrix} cannot be directly used as $[B|T]$ due to the improper number of columns and the weight-4 column, but it can be easily modified to meet the requirements for $[B|T]$ by splitting the weight-4 column into two weight-2 columns. The construction procedure of $[B|T]$ for $J \equiv 0 \mod 6$, $J \geq 12$, and $d_c=J/2$ is summarized as:

	\begin{enumerate}
		\item[1.] Construct PBD$(J-1,\{ 3,5^\star \})$.
		\item[2.] Remove a row of weight $(J-4)/2$ and its incident columns except the weight-5 column from an incidence matrix of PBD$(J-1,\{ 3,5^\star \})$.
		\item[3.] Split the weight-4 column into two weight-2 columns and move them to the leftmost to obtain $[B|T]$.
	\end{enumerate}

\begin{figure}[tb]
	\centering
	\subfigure[An incidence matrix of PBD$(11,\{ 3,5^\star \})$.]{\includegraphics[scale=0.8]{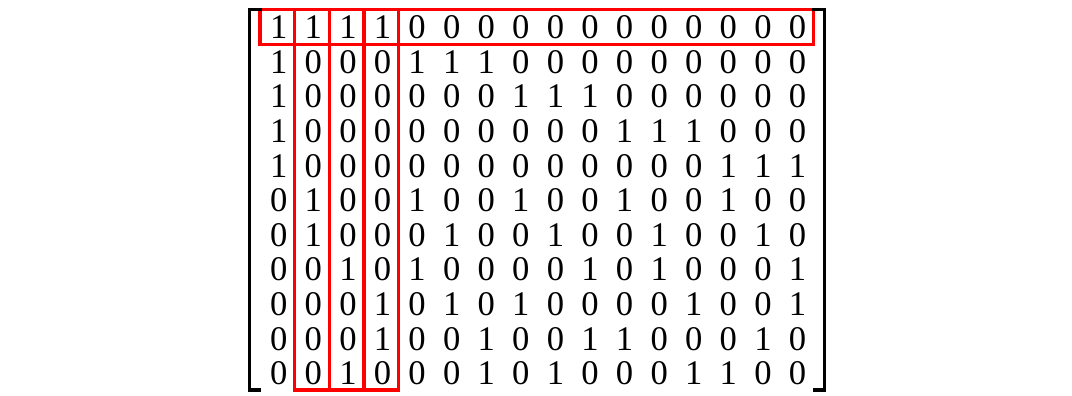}%
	\label{fig:12_24_const1}}\\
%	\subfigure[The incidence matrix of OP$(10)$.]{\includegraphics[scale=0.8]{12_24_const2}%
%	\label{fig:12_24_const2}}\\
	\subfigure[A $12 \times 24$ regular protograph with $d_v=3$ and $n_2=10$.]{\includegraphics[scale=0.8]{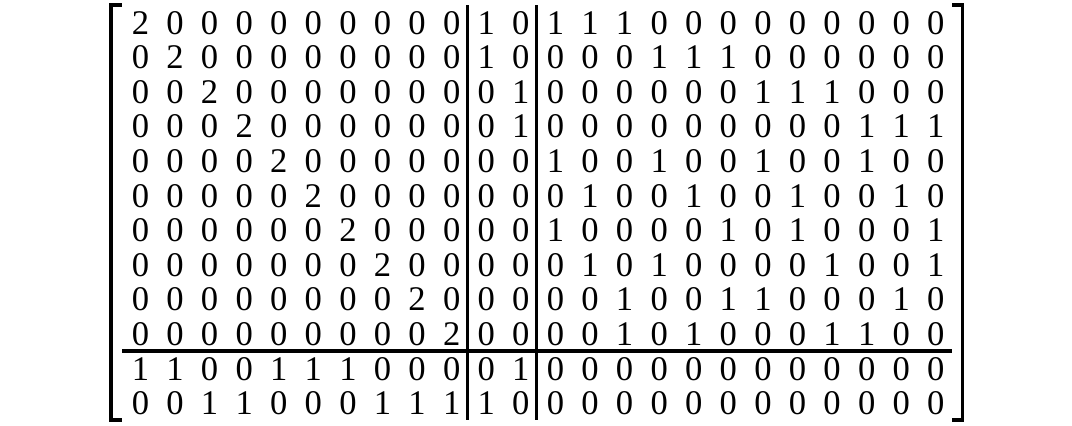}%
	\label{fig:12_24_proto}}
	\caption{The construction of a $12 \times 24$ regular protograph with $d_v=3$ and $n_2=10$.}
	\label{fig:12_24}
\end{figure}

\vspace{2mm}		
\begin{example}
The construction process for a $12 \times 24$ regular protograph with $d_v=3$ and $n_2=10$ is illustrated in Fig. \ref{fig:12_24}. An incidence matrix of PBD$(11,\{ 3,5^\star \})$ is shown in Fig. \ref{fig:12_24_const1}. We can see that the submatrix consisting of the columns incident to the first row has exactly one 1 in each row except the first row. By removing the first row and the second, the third, and the fourth columns and splitting the weight-4 column into two weight-2 column, $[B|T]$ is obtained. The resulting $12 \times 24$ regular protograph with $d_v=3$ and $n_2=10$ is shown in Fig. \ref{fig:12_24_proto}.
\end{example}

\vspace{2mm}						
\indent \textit{6)} $J=9$, $d_c=4$ and $J=10$, $d_c=6$:

\vspace{2mm}
There only remain two cases to provide the construction methods of all regular protographs in Theorem \ref{theorem:main}. When $J=9$ and $d_c=4$, we have $L_G=1$ and $L_T=4$, and $[B|T]$ is a $7 \times 5$ matrix with row-weight 2.
Although the construction method of $[B|T]$ for $J \equiv 3 \mod 6$ and $d_c = (J-1)/2$ cannot be directly used, we can construct $[B|T]$ from an incidence matrix of $S(2,3,7)$ in Fig. \ref{fig:9_12_const}.
Since any two columns of an incidence matrix of $S(2,3,7)$ have a common 1, removing the first two columns from an incidence matrix results in a $7 \times 5$ matrix where one row has the weight 1, four rows have the weight 2, and the remaining two rows have the weight 3 as shown in Fig. \ref{fig:9_12_const}. 
To obtain $[B|T]$, first delete 1 from each of two rows of weight 3 in the $7 \times 5$ matrix such that two deleted 1's do not belong to the same column and the columns containing two deleted 1's do not have 1 in the row of weight 1. These two deleted 1's are marked by circle in Fig. \ref{fig:9_12_const}. Then, by replacing a 0 at the row of weight 1 and one of the columns containing the deleted 1's with a 1, $[B|T]$ is constructed and the resulting $9 \times 12$ regular protograph is shown in Fig. \ref{fig:9_12_proto}.

When $J=10$ and $d_c=6$, we have $L_G=4$ and $L_T=8$, and $B$ has disjoint four column-wise pairs of 1's.
An incidence matrix of a symmetric configuration $8_3$ can be used as $T$, which does not have disjoint four column-wise pairs of 1's.
A $10 \times 20$ regular protographs with $d_v=3$ and $n_2=8$ is shown in Fig. \ref{fig:10_20}.

%\begin{figure}[tb]
%    \centering
%	\includegraphics[scale=0.8]{10_20_proto}%
%    \caption{A $10 \times 20$ regular protograph with $d_v=3$ and $n_2=8$.}
%    \label{fig:10_20}
%\end{figure}

\begin{figure}[tb]
	\centering
	\subfigure[An incidence matrix of $S(2,3,7)$.]{\includegraphics[scale=0.8]{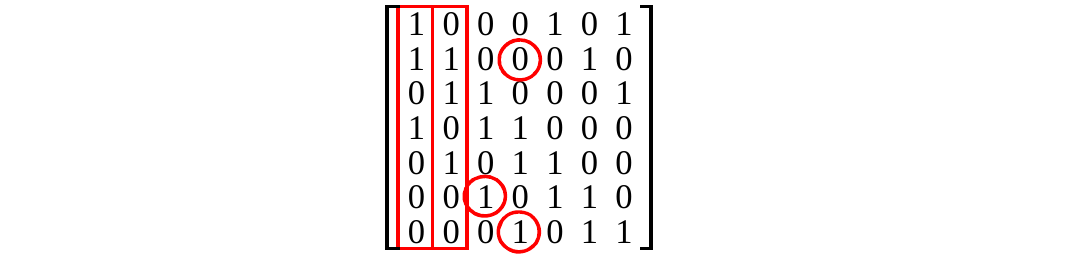}%
	\label{fig:9_12_const}}\\
	\subfigure[A $9 \times 12$ regular protograph with $d_v=3$ and $n_2=7$.]{\includegraphics[scale=0.8]{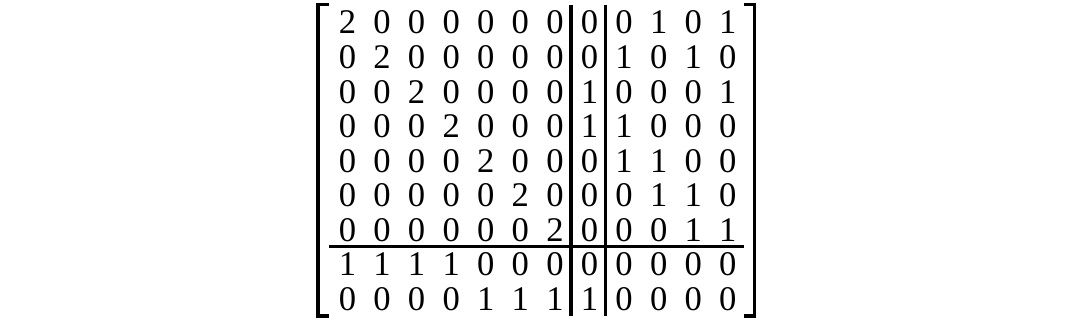}%
	\label{fig:9_12_proto}}\\
	\caption{The construction of a $9 \times 12$ regular protograph with $d_v=3$ and $n_2=7$.}
	\label{fig:9_12}
\end{figure}		

\begin{figure}[tb]
    \centering
	\includegraphics[scale=0.8]{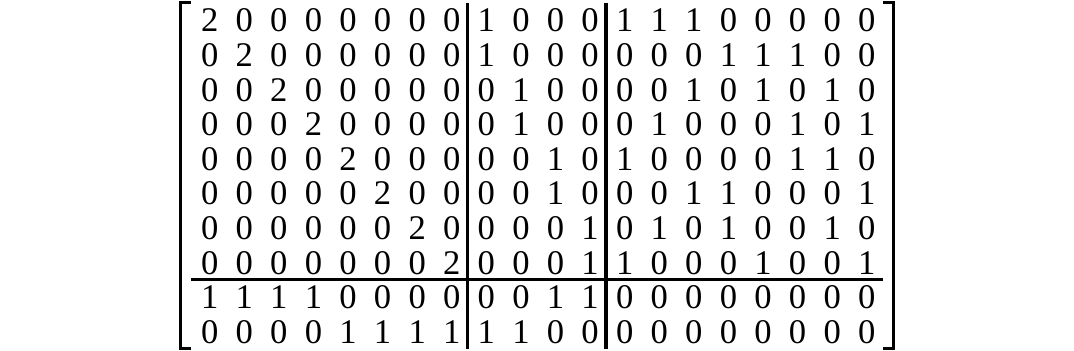}%
    \caption{A $10 \times 20$ regular protograph with $d_v=3$ and $n_2=8$.}
    \label{fig:10_20}
\end{figure}

\subsection{Regular Protographs With $d_v=3$ and $n_2 < J-2$, and With $d_v \geq 4$}
\label{subsec:other}

Regular protographs which do not induce inevitable cycles of length less than 14 for the case of $d_v=3$ and $n_2 < J-2$ and the case of $d_v \geq 4$ also have the same structure in Fig. \ref{fig:structure} and the construction method in the previous subsection can be similarly applied to these cases.
However, we do not elaborate on deriving necessary conditions like Theorem \ref{theorem:main} and providing specific construction methods for all cases because they should be done case by case and are very lengthy.
Instead, for given $J$, $L$, $d_v$, $d_c$, and $n_2$, we provide a general framework for checking the constructibility and constructing each submatrix.

First, some basic conditions on the parameters $J$, $L$, $d_v$, $d_c$, and $n_2$ are provided to determine whether a regular protograph with the given parameters can be potentially constructed. In $F$, the number of all possible column-wise pairs of 1's should be larger than or equal to the number of actual column-wise pairs of 1's, that is, ${J-n_2 \choose 2} \geq n_2 {d_v -2 \choose 2}$. Also, the last $J-n_2$ rows must have $(J-n_2)d_c$ 1's and the matrix $F$ must have $n_2 (d_v -2)$ 1's, and thus we have $(J-n_2)d_c \geq n_2 (d_v -2)$.

Second, consider constructing $[B|T]$ of size $n_2 \times (L-n_2)$. The matrix $[B|T]$ has the constant row-weight $d_c-2$ and does not have any repeated column-wise pairs of 1's to avoid the second and the third ICI subgraphs of $\mathcal{P}_{12}$ in $[A|B|T]$, and $T$ has the constant column-weight $d_v$. The matrix $[B|T]$ can be constructed from an incidence matrix of block designs such as $S(2,k,v)$, configurations $(v_r,b_k)$, PBD$(v,K)$, group divisible designs (GDD) \cite{Colbourn1} and so on because they do not have any repeated column-wise pairs of 1's. If an incidence matrix of an $S(2,k,v)$ or a configuration $(v_r,b_k)$ has the desired size of $[B|T]$, it can be directly used as $[B|T]$ where there is actually no $B$. Otherwise, an incidence matrix of some block designs can be used as $[B|T]$ by doing a simple modification. Note that an incidence matrix of PBDs or GDDs may not have a constant row-weight while an incidence matrix of $S(2,k,v)$ or configurations $(v_r,b_k)$ is always regular.

To obtain $[B|T]$ from an incidence matrix of different size, we may use the following modification schemes:
\begin{enumerate}
	\item[1)] Remove some rows.
	\item[2)] Remove a row and some columns incident to the row.
	\item[3)] Remove some parallel classes.
	\item[4)] Delete some 1's and insert some columns not to have any repeated column-wise pairs of 1's.
	\item[5)] Insert some parallel classes not to have any repeated column-wise pairs of 1's.
\end{enumerate}
By properly applying these modification schemes, an incidence matrix is changed into $[B|T]$ having no repeated column-wise pairs of 1's and proper size. Moreover, Schemes 1), 3), and 5) change all row-weights by the same amount and Schemes 2) and 4) flexibly control row-weights according to how to select columns and 1's. Therefore, we can freely use the above modification schemes until the desired size and the constant row-weight $d_c -2$ of $[B|T]$ are achieved.

For given $J$, $L$, $d_v$, $d_c$, and $n_2$, it may be possible for $[B|T]$ to take various forms, which implies that each $B$ may have a different number of columns and a different distribution of 1's.
Therefore, some bounds on $L_G$, or the number of columns in $B$, need to be derived to construct $[B|T]$. Since the number of 1's in $G$ is $(J-n_2) d_c - n_2 (d_v -2)$ and each column in $G$ can have weight from $1$ to $d_v$, we have $L_G \leq (J-n_2) d_c - n_2 (d_v -2) \leq d_v L_G$ which yields $\{ (J-n_2) d_c - n_2 (d_v -2) \} /d_v \leq L_G \leq (J-n_2) d_c - n_2 (d_v -2)$.
Therefore, when $[B|T]$ is constructed by selecting and modifying an incidence matrix of a block design, the above bound on $L_G$ must be considered.

Lastly, consider constructing $F$ of size $(J-n_2) \times n_2$ and $G$ of size $(J-n_2) \times L_G$. Column-weights of $G$ are already determined if $B$ is designed and 1's in $G$ should be located to avoid the second and the third ICI subgraphs of $\mathcal{P}_{12}$ in the union of $A$, $B$, and $G$. Then, for a given $G$, $(J-n_2)d_c$ 1's in $F$ should be located such that the union of $A$, $F$, and $G$ does not contain the second and the third ICI subgraphs of $\mathcal{P}_{12}$, and the union of $A$, $B$, $F$, and $G$ does not contain $\mathcal{P}_{10}$ while enforcing row-weights of $[F|G]$ to be $d_c$ and column-weights of $F$ to be $d_v -2$.

For given parameters $J$, $L$, $d_v$, $d_c$, and $n_2$, a general procedure for constructing regular protographs which avoid inevitable cycles of length less than 14 is summarized as:
\begin{enumerate}
	\item[1.] Check if the parameters satisfy the conditions ${J-n_2 \choose 2} \geq n_2 {d_v -2 \choose 2}$ and $(J-n_2)d_c \geq n_2 (d_v -2)$. If the conditions are not satisfied, stop the procedure.
	\item[2.] Obtain $L_G$ satisfying $\{ (J-n_2) d_c - n_2 (d_v -2) \} /d_v \leq L_G \leq (J-n_2) d_c - n_2 (d_v -2)$.
	\item[3.] Construct $[B|T]$ using $L_G$ obtained in Step 2 from an incidence matrix of a proper block design.
	\item[4.] Construct $[F|G]$ satisfying the weight constraints such that the union of $A$, $B$, $F$, and $G$ does not have $\mathcal{P}_{10}$ and $\mathcal{P}_{12}$ as its subgraph.
\end{enumerate}

\vspace{2mm}
\begin{example}
Consider the construction of a $15 \times 30$ regular protograph with $d_v=3$, $d_c=6$, and $n_2=12$. 
The given parameters satisfy the conditions ${J-n_2 \choose 2} \geq n_2 {d_v -2 \choose 2}$ and $(J-n_2)d_c \geq n_2 (d_v -2)$.
The $12 \times 18$ matrix $[B|T]$ has the row-weight 4 and $2 \leq L_G \leq 6$.
An incidence matrix of a symmetric configuration $12_3$ is chosen for the construction of $[B|T]$, which is constructed by removing two parallel classes from a $12 \times 20$ incidence matrix of the configuration $(12_5,20_3)$ in Fig. \ref{fig:14_28_const2}.
By inserting a parallel class consisting of six weight-2 columns to an incidence matrix of the symmetric configuration $12_3$, $[B|T]$ with $L_G =6$ is constructed, where any repeated column-wise pairs of 1's do not appear.
Since the column-weight of $B$ is 2, the $3 \times 6$ matrix $G$ should have the column-weight 1.
Let each row of $G$ have two 1's. Then $F$ should have the column-weight 1 and the row-weight 4, and the 1's in $F$ can be properly distributed so that $\mathcal{P}_{10}$ and $\mathcal{P}_{12}$ do not appear in the union of $A$, $B$, $F$, and $G$.
The resulting $15 \times 30$ regular protograph with $d_v=3$, $d_c=6$, and $n_2=12$ is shown in Fig. \ref{fig:15_30}.
\end{example}

\begin{figure}[tb]
    \centering
	\includegraphics[scale=0.8]{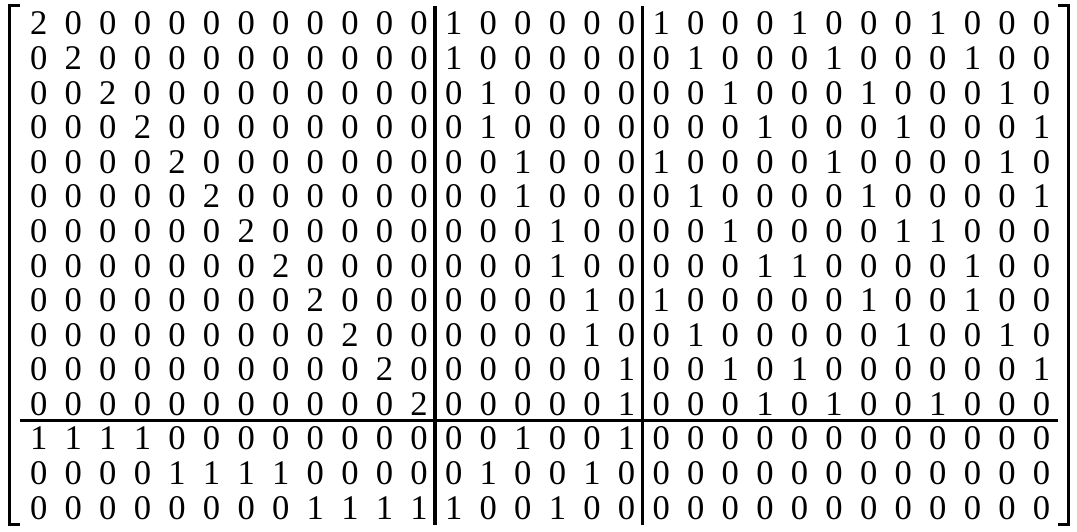}%
    \caption{A $15 \times 30$ regular protograph with $d_v=3$ and $n_2=12$.}
    \label{fig:15_30}
\end{figure}

\vspace{2mm}
\begin{example}
Consider the construction of a $28 \times 49$ regular protograph with $d_v=4$, $d_c =7$, and $n_2=21$.
We can check that those parameters satisfy two necessary conditions for the construction and $2 \leq L_G \leq 7$. For constructing a $21 \times 28$ $[B|T]$, an incidence matrix of a symmetric configuration $21_4$ is considered. We can find a parallel class consisting of seven weight-3 columns such that if those seven columns are inserted to the incidence matrix, any repeated column-wise pairs of 1's still do not appear. Thus, we obtain $[B|T]$ with $L_G=7$. The $7 \times 7$ matrix $G$ should have the column-weight 1 and its row-weight can be set to 1. The matrix $F$ has the size $7 \times 21$, the column-weight 2, and the row-weight 6. Due to ${7 \choose 2}=21$, $F$ can be constructed not to have any repeated column-wise pairs of 1's. Also, we can make $F$ to avoid $\mathcal{P}_{10}$ in the union of $A$, $B$, $F$, and $G$. The resulting $28 \times 49$ regular protograph with $d_v=4$, $d_c=7$, and $n_2=21$ is shown in Fig. \ref{fig:28_49}.
\end{example}

\begin{figure*}[tb]
    \centering
	\includegraphics[scale=0.75]{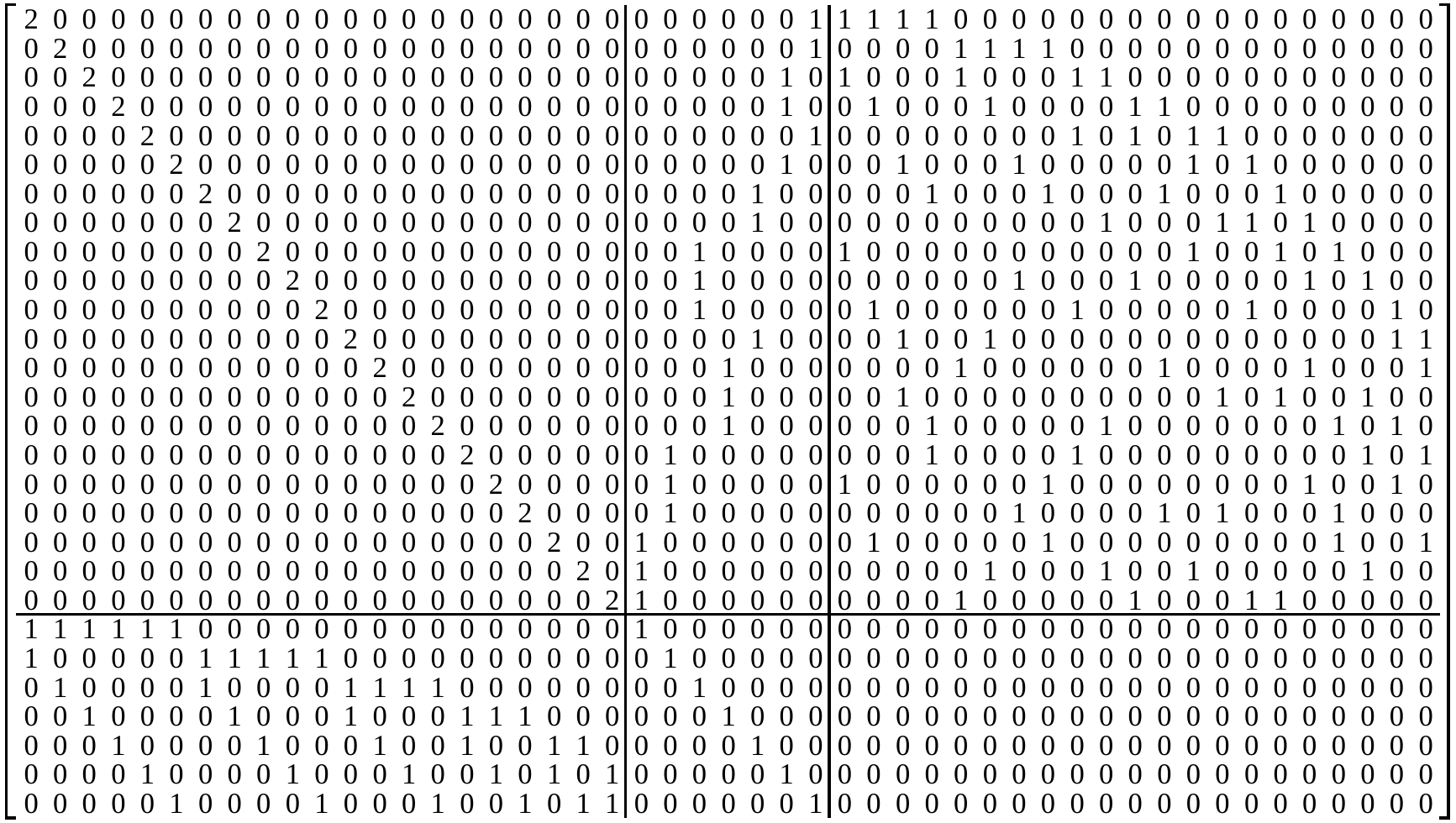}%
    \caption{A $28 \times 49$ regular protograph with $d_v=4$ and $n_2=21$.}
    \label{fig:28_49}
\end{figure*}

%\begin{figure}[tb]
%    \centering
%	\includegraphics[scale=0.7]{20_25_proto}%
%    \caption{A $20 \times 25$ regular protograph with $d_v=4$, $n_2=14$.}
%    \label{fig:20_25}
%\end{figure}

\vspace{2mm}
\section{Construction of QC LDPC Codes and Their Minimum Hamming Distances}\label{sec:dmin}

\subsection{Construction of QC LDPC Codes From the Proposed Protographs}

To verify the effectiveness of the proposed protographs, QC LDPC codes will be constructed by determining the lift size and assigning an appropriate shift value to each edge of the protographs. Given a protograph, it is not easy to find all shift values even for a moderate lift size such that the girth of the QC LDPC code is the same as the length of the shortest inevitable cycle. Huang \textit{et al.} \cite{Huang} proposed a search algorithm for small lift size and a shift value assignment scheme to achieve the target girth based on greedy search. This algorithm is originally designed for single-edge protographs. However, by a slight modification, this algorithm can be extended to the case of multiple-edge protographs.

Consider a $J \times L$ protograph $P$ with the column-weight $d_v$, the row-weight $d_c$, and the lift size $z$. Each column of $P$ has $d_v$ shift values and let $s_{l,i}$, $i=0,\ldots,d_v-1$, denote the $i$-th shift value of the $l$-th column in $P$. Our goal is to determine all shift values $s_{l,i}$ and search the minimum $z$ when a protograph and a target girth $g$ of QC LDPC codes are given. Let $\mathcal{W}_n$ denote the set of all TNC walks of length $n$ in $P$. Then, by Lemma \ref{lemma:cycle}, the condition for achieving the target girth $g$ of QC LDPC codes is that for any $W \in \mathcal{W}_n$, $n=4,6,\ldots,g-2$, the shift sum $s(W)$ satisfies $s(W) \neq 0 \mod z$. However, it requires too much computational complexity to find $s_{l,i}$ and the minimum $z$ satisfying the above condition by considering all search space of $s_{l,i}$ and $z$.

In order to reduce the search space of shift values, let $s_{l,i} = r_{i} m_l$ as in \cite{Huang}, where $r_i$ is the $(i+1)$-st element of the set $\{ 0,1,3,7,12,20, \ldots \}$ which is constructed from $\{ r_0 =0\}$ by adding $r_i$, $i=1,2,\ldots$, in order such that $r_i = r_{i-1} + \min_{j,k<i}{\left[ \mathbb{N} \setminus \{ |r_j - r_k| \} \right]}$. Thus we only need to find $L$ values of $m_l$ instead of $d_v L$ values of $s_{l,i}$. Moreover, for further reduction of computational complexity, $m_l$ is determined in a greedy manner, that is, shift values of the $l$-th column in $P$ are determined by considering only the first $l$ columns in $P$. For this, let $\mathcal{W}_n^{(l)}$ denote the set of all TNC walks of length $n$ in the matrix consisting of the first $l$ columns of $P$.

For a given target girth $g$, if $s_{l,i}$ is already determined such that $s(W) \neq 0$ for any $W \in \mathcal{W}_n$, $n=4,6,\ldots,g-2$, the minimum $z$, denoted by $z_{\min}$, can be sub-optimally determined as $z_{\min} = \max \left\{ |s(W)| ~|~ W \in \mathcal{W}_n,~n= 4,6,\ldots,g-2 \right\} +1$. Note that for any $z \geq z_{\min}$, the target girth is achieved.

\begin{figure}[tb]
\noindent\fbox{\parbox{0.98\linewidth}{%\small
\textbf{Algorithm 1: Greedy Search for the Minimum Lift Size and Shift Values}

\textbf{INPUT}: Target girth $g$, $J \times L$ protograph, search bound $\Gamma_{\max}$

\textbf{OUTPUT}: $m_l$ ($0 \leq l \leq L-1$) and $z_{\min}$

\textbf{INITIALIZATION}: $r_0=0$, $r_i = r_{i-1} + \min_{j,k<i}{\left[ \mathbb{N} \setminus \{ |r_j - r_k| \} \right]}$ for $1 \leq i \leq d_v -1$

\textbf{MAIN ROUTINE}

		\begin{itemize}
			\item[] \textbf{for} $l=0$ to $L-1$ \textbf{begin}
			\begin{itemize}
				\item[] \textbf{for} $m_{l}=-\Gamma_{\max}$ to $\Gamma_{\max}$ \textbf{begin}
				\begin{itemize}
					\item[] Let $s_{l,i} = r_{i} m_l$ for $0 \leq i \leq d_v-1$.
					\item[] If $s(W) \neq 0$ for any $W \in \mathcal{W}_n^{(l)}$, $n=4,6,\ldots,g-2$,
					\item[] \hspace{3mm} $z_{\min}^{(l)}(m_l) = \max \left\{ |s(W)| ~|~ W \in \mathcal{W}_n^{(l)},~n= 4,6,\ldots,g-2 \right\} +1$.
					\item[] Otherwise, $z_{\min}^{(l)}(m_l) = \infty$.
				\end{itemize}
				\item[] \textbf{end}
				\item[] Select the minimum $z_{\min}^{(l)}(m_l)$ and save the minimum $z_{\min}^{(l)}(m_l)$ to $z_{\min}^{(l)}$ and also save the argument to $m_l$.
				\item[] If there are multiple minimums, randomly pick any one.
			\end{itemize}
			\item[] \textbf{end}
			\item[] $z_{\min} = z_{\min}^{(L-1)}$
		\end{itemize}

}}
\end{figure}

An algorithm to construct QC LDPC codes of moderate length by determining all shift values and searching the minimum lift size, called Algorithm 1, is provided as follows.
If the target girth $g$ is set to the length of the shortest inevitable cycle, we can generate QC LDPC codes of moderate length with the maximum achievable girth from the proposed protographs.
Note that the computational complexity of Algorithm 1 is the same for both single-edge protographs and multiple-edge protographs under the same parameter values.

Four QC LDPC codes are generated by using Algorithm 1.
From the $9 \times 15$ protograph in Fig. \ref{fig:9_15}, a $(15000,6000)$ QC LDPC code with girth 14, denoted by Proposed Code 1, is constructed, which has $z=1000$ and $\{ m_l \} = \{ -105, 36, 45, 75, -69, -303, -393, 127, -31, -199, 200, 184, 86, 200,$ $ 199 \}$. From the $9 \times 12$ protograph in Fig. \ref{fig:9_12_proto}, a $(3600,900)$ QC LDPC code with girth 14, denoted by Proposed Code 2, is constructed, which has $z=300$ and $\{ m_l \} = \{ -12, 18, -39, 75, -57,$ $ 120, 15, 17, 0, -6, -8, -8 \}$.
From the $6 \times 12$ protograph in Fig. \ref{fig:6_12}, a $(7200,3600)$ QC LDPC code with girth 12, denoted by Proposed Code 3, is constructed, which has $z=600$ and $\{ m_l \} = \{ -93, 7, 47, -52, -29, -192, 30, 29, 30, 3, 19, 42 \}$. From the $6 \times 8$ protograph in Fig. \ref{fig:6_8}, a $(800,200)$ QC LDPC code with girth 12, denoted by Proposed Code 4, is constructed, which has $z=100$ and $\{ m_l \} = \{ -3, 85, -18, -6, -7, 2, -5, -15 \}$.

\subsection{Upper Bounds on the Minimum Hamming Distance of the Proposed QC LDPC Codes}

Smarandache and Vontobel \cite{Smarandache} derived two upper bounds on the minimum Hamming distance of QC LDPC codes. While one bound needs whole code specifications, e.g., the structure of the protograph, the lift size, and the shift values, the other bound only requires knowledge of the protograph. 

These two upper bounds are shown in Theorems \ref{theorem:dmin_1} and \ref{theorem:dmin}, and they are directly derived by finding some low-weight codewords as in Lemma \ref{lemma:dmin}.
Let $Q_{\mathcal{S}}$ denote the submatrix of $Q$ that contains only the columns of $Q$ whose index appears in the set $\mathcal{S}$.

\vspace{2mm}
\begin{definition}[\cite{Smarandache}]
The permanent of an $m \times m$ matrix $Q = [q_{i,j}]$ over some commutative ring is defined to be
	\begin{equation*}
		\mathrm{perm} (Q) := \sum_{\sigma} \prod_{i \in \{ 0,\ldots,m-1 \}}q_{i,\sigma (i)}
	\end{equation*}
where the summation is over all $m!$ permutations $\sigma$ on the set $\{ 0,\ldots,m-1 \}$.
\end{definition}

\vspace{2mm}
\begin{lemma}[\cite{Smarandache}]
Let $\mathcal{C}$ be a binary QC LDPC code defined by a $J \times L$ polynomial matrix $H(x)$ with the lift size $z$. Let $\mathcal{S}$ be an arbitrary size-$(J+1)$ subset of $\{ 0,1,\ldots,L-1 \}$ and let $c(x)=[ c_0(x),c_1(x),\ldots,c_{L-1}(x)]$, where $c_i(x)$ is a polynomial over $\mathbb{F}_2(x)/(x^z+1)$ defined by
\begin{equation*}
	c_i(x) = \begin{cases} \mathrm{perm} \left( H_{\mathcal{S}\setminus \{ i \} }(x) \right), & \mathrm{if}~i \in \mathcal{S} \\ 0, & \mathrm{otherwise}. \end{cases}
\end{equation*}
Then $c(x)$ is a codeword of $\mathcal{C}$.
\label{lemma:dmin}
\end{lemma}

%\vspace{2mm}
%An upper bound on the minimum Hamming distance of a QC LDPC code for which only its protograph is specified is given in the following theorem.

\vspace{2mm}
\begin{theorem}[\cite{Smarandache}]
Let $\mathcal C$ be a binary QC LDPC code defined by a $J \times L$ polynomial matrix $H(x)$ with the lift size $z$. Then the minimum Hamming distance of $\mathcal{C}$ is upper bounded as
	\begin{equation}
		d_{\min} (\mathcal C) \leq \min^{*}_{\mathcal{S} \subseteq \{ 0,\ldots,L-1 \} \atop |\mathcal{S}|=J+1} \sum_{i \in \mathcal{S}} \mathrm{wt} \left( \mathrm{perm} \left( H_{\mathcal{S} \setminus \{ i \}}(x) \right) \right)
	\label{eq:dmin1}
	\end{equation}
where the operator $\stackrel{*}{\mathrm{min}}$ gives back the minimum value of all nonzero entries in a list of values.
\label{theorem:dmin_1}
\end{theorem}

\vspace{2mm}
\begin{theorem}[\cite{Smarandache}]
Let $\mathcal C$ be a binary QC LDPC code lifted from a $J \times L$ protograph $P$. Then the minimum Hamming distance of $\mathcal{C}$ is upper bounded as
	\begin{equation}
		d_{\min} (\mathcal C) \leq \min^{*}_{\mathcal{S} \subseteq \{ 0,\ldots,L-1 \} \atop |\mathcal{S}|=J+1} \sum_{i \in \mathcal{S}} \mathrm{perm} \left( P_{\mathcal{S} \setminus \{ i \}} \right).
	\label{eq:dmin2}
	\end{equation}
\label{theorem:dmin}
\end{theorem}

\vspace{2mm}
Theorems \ref{theorem:dmin_1} and \ref{theorem:dmin} imply that for given $J$, $L$, $d_v$, and $d_c$, these two upper bounds on the minimum Hamming distance of QC LDPC codes possibly increase as the number of multiple edges in the protograph increases, which is supported by examples for some regular protographs in \cite{Smarandache}. Note that the bound in (\ref{eq:dmin2}) is not tighter than the bound in (\ref{eq:dmin1}), but the former approaches the latter for a large $z$ and proper shift values.

\begin{figure}[tb]
	\centering
	\subfigure[$9 \times 15$]{\includegraphics[scale=0.8]{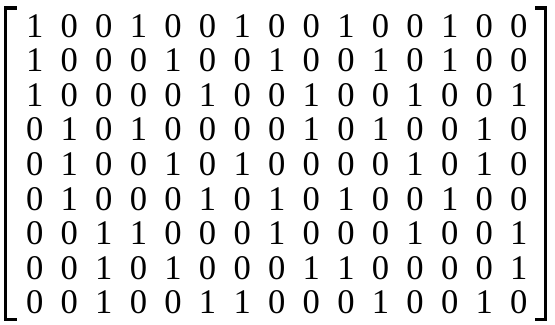}%
	\label{fig:9_15_dmin}} \hspace{5mm}
	\subfigure[$9 \times 12$]{\includegraphics[scale=0.8]{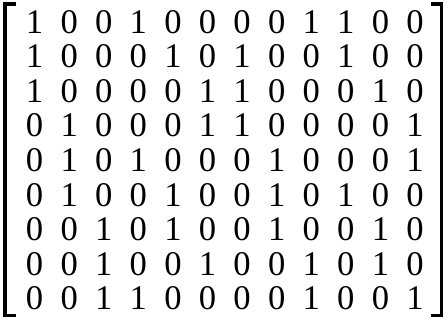}%
	\label{fig:9_12_dmin}}
	\caption{Two single-edge regular protographs with $d_v=3$ avoiding inevitable cycles of length $<$ 14.}
	\label{fig:dmin_14}
\end{figure}

Consider the (15000,6000) Proposed Code 1. The upper bounds in (\ref{eq:dmin1}) and (\ref{eq:dmin2}) for this code are 246 and 256, respectively. For comparison, a QC LDPC code with the same parameter values is generated from the $9 \times 15$ single-edge regular protograph in Fig. \ref{fig:9_15_dmin} by using Algorithm 1. This single-edge protograph is constructed by attaching the last three columns to an incidence matrix of $S(2,3,9)$ to avoid inevitable cycles of length less than 14. The upper bounds in (\ref{eq:dmin1}) and (\ref{eq:dmin2}) for this code are 218 and 230, respectively.

Consider the (3600,900) Proposed Code 2. The upper bounds in (\ref{eq:dmin1}) and (\ref{eq:dmin2}) for this code are 362 and 416, respectively. For comparison, a QC LDPC code with the same parameter values is generated from the $9 \times 12$ single-edge regular protograph in Fig. \ref{fig:9_12_dmin} by using Algorithm 1. By using the construction method in \cite{Kim}, this single-edge protograph is constructed by concatenating an incidence matrix of a $(9_2,6_3)$ configuration and cyclically row-shifted matrix of it. The upper bounds in (\ref{eq:dmin1}) and (\ref{eq:dmin2}) for this code are 314 and 384, respectively.

Consider the (7200,3600) Proposed Code 3. The upper bounds in (\ref{eq:dmin1}) and (\ref{eq:dmin2}) for this code are all 68. For comparison, a QC LDPC code with the same parameter values is generated from the $6 \times 12$ single-edge regular protograph in Fig. \ref{fig:6_12_dmin} by using Algorithm 1. This single-edge protograph is the best one of randomly constructed protographs in the sense of upper bounds on the minimum Hamming distance. The upper bounds in (\ref{eq:dmin1}) and (\ref{eq:dmin2}) for this code are all 56.

Finally, consider the (800,200) Proposed Code 4. The upper bounds in (\ref{eq:dmin1}) and (\ref{eq:dmin2}) for this code are 130 and 174, respectively. For comparison, a QC LDPC code with the same parameter values is generated from the $6 \times 8$ single-edge regular protograph in Fig. \ref{fig:6_8_dmin} by using Algorithm 1. This single-edge protograph is the best one of randomly constructed protographs in the sense of upper bounds on the minimum Hamming distance. The upper bounds in (\ref{eq:dmin1}) and (\ref{eq:dmin2}) for this code are 98 and 110, respectively.

\begin{figure}[tb]
	\centering
	\subfigure[$6 \times 12$]{\includegraphics[scale=0.8]{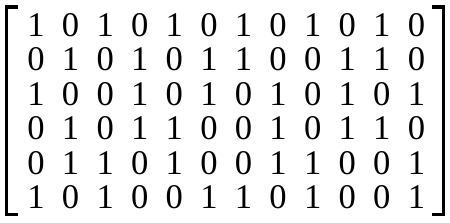}%
	\label{fig:6_12_dmin}} \hspace{5mm}
	\subfigure[$6 \times 8$]{\includegraphics[scale=0.8]{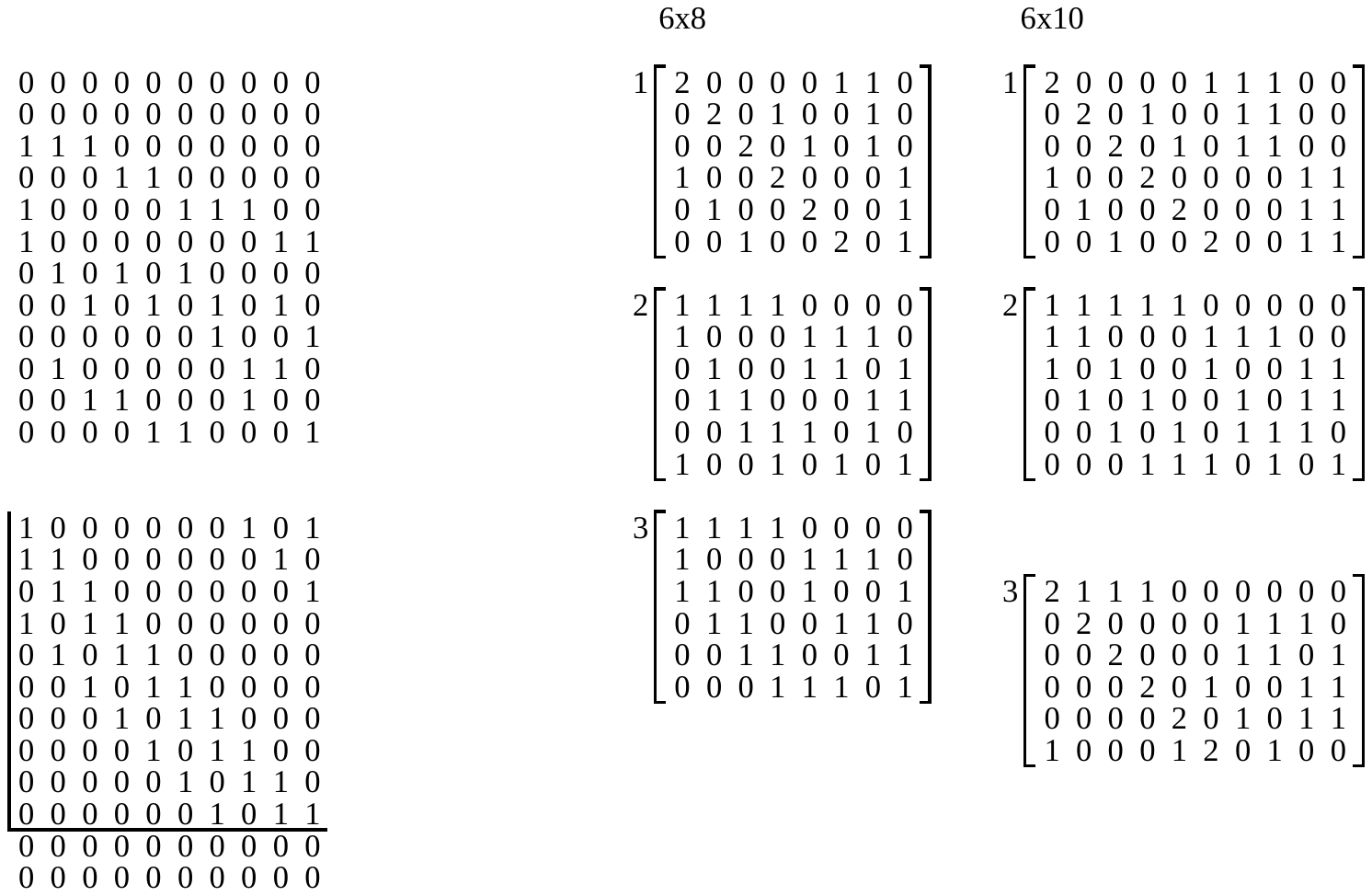}%
	\label{fig:6_8_dmin}}
	\caption{Two single-edge regular protographs with $d_v=3$ avoiding inevitable cycles of length $<$ 12.}
	\label{fig:dmin_12}
\end{figure}

The above results clearly show that two upper bounds (\ref{eq:dmin1}) and (\ref{eq:dmin2}) on the minimum Hamming distance of QC LDPC codes are affected in a positive way by using double edges in the protographs. In general, a multiple-edge protograph is more difficult to design than a single-edge protograph under the condition that they induce the shortest inevitable cycles of the same length. However, if multiple-edge protographs are once constructed, QC LDPC codes lifted from them can potentially give a larger upper bound on the minimum Hamming distance than those lifted from single-edge protographs.

\subsection{Comparison of Error Correcting Performance}

Performance of four proposed QC LDPC codes, that is, Proposed Code 1 to 4 is compared with those of the progressive edge-growth LDPC codes, called PEG 1 to 4 \cite{Hu} and the QC LDPC codes, called PEG QC 1 to 4 \cite{Li} with the same code length, code rate, and column-weight.
PEG LDPC codes and PEG QC LDPC codes are well known to have good error correcting performance comparable to those of random LDPC codes. Note that the girths of such $(15000,6000)$, $(3600,900)$, $(7200,3600)$, and $(800,200)$ PEG LDPC codes and PEG QC LDPC codes are 12, 12, 12, and 10, respectively, and these codes are obtained by the PEG algorithm to have as large girth as possible.

Simulation is carried out through the binary input additive white Gaussian noise (BIAWGN) channel.
The belief propagation (BP) decoding algoirthm is used and the number of maximum iterations is set to 100. 
The frame error rate (FER) performances of all the above LDPC codes are compared in Fig. \ref{fig:FER} and we can see that the proposed QC LDPC codes show as good error correcting performance as the PEG LDPC codes and the PEG QC LDPC codes. Note that the bit error rate (BER) curves behave qualitatively the same as the FER curves and they are omitted in this paper.

\begin{figure*}[tb]
    \centering
	\includegraphics[scale=1.2]{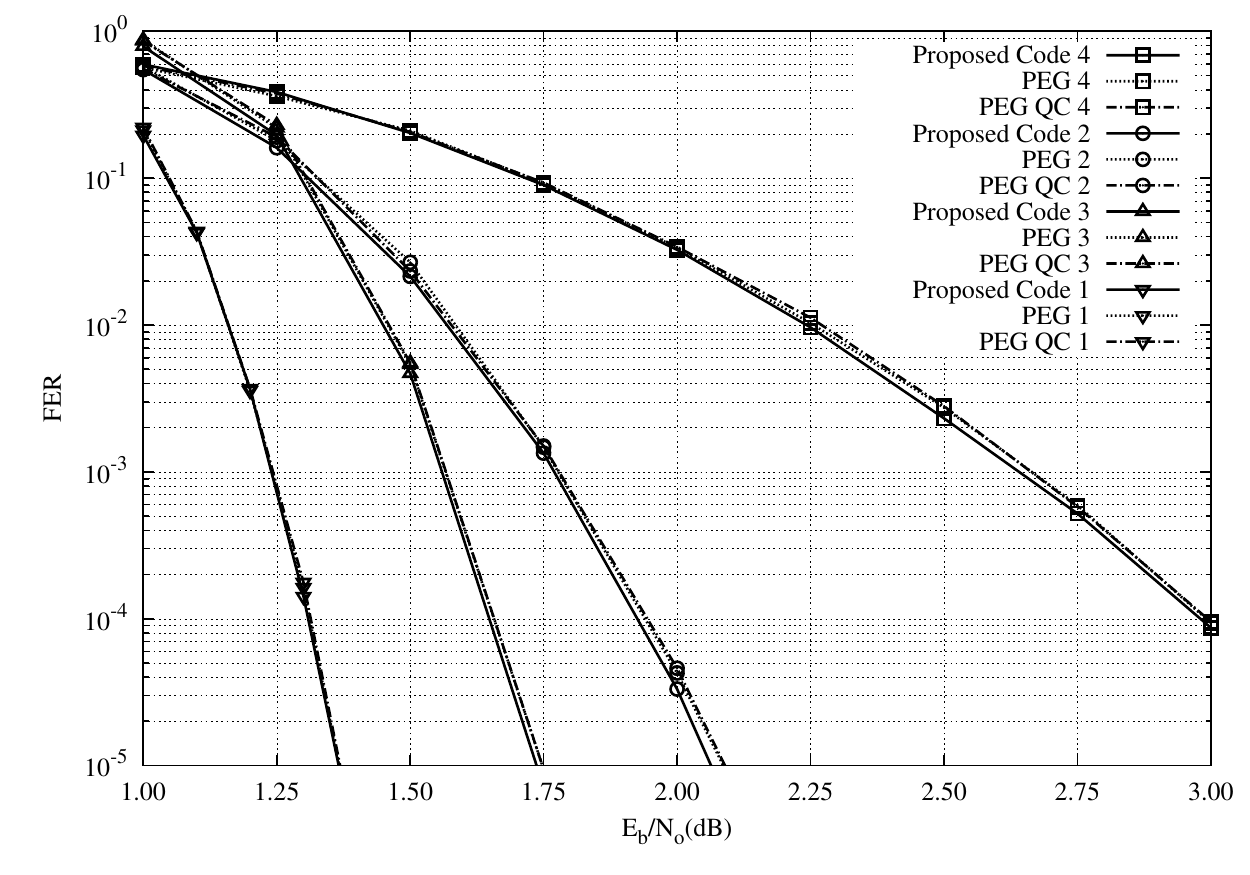}%
    \caption{Error correcting performance comparison of the proposed QC LDPC codes, the PEG LDPC codes, and the PEG QC LDPC codes.}
    \label{fig:FER}
\end{figure*}

%\begin{figure*}[tb]
%    \centering
%	\includegraphics[scale=1.2]{BER}%
%    \caption{BER performance comparison of the proposed QC LDPC codes, the PEG LDPC codes and the PEG QC LDPC codes.}
%    \label{fig:BER}
%\end{figure*}

\vspace{2mm}
\section{Conclusions}\label{sec:conclusion}

The subgraphs of protographs, which cause inevitable cycles in the QC LDPC codes, are fully investigated in allowance with multiple edges through the graph-theoretic approach. For regular QC LDPC codes with girth larger than or equal to 12, we propose a systematic construction method of protographs which avoid inevitable cycles of length less than 12 by using balanced ternary designs. For regular QC LDPC codes with girth larger than or equal to 14, we provide construction methods of all $J \times L$ protographs with column-weight three and the number of double edges $J-2$ by using various block designs. These construction methods can be extended to construct regular protographs with smaller number of double edges and with column-weight larger than three. Also, a construction algorithm of QC LDPC codes from the proposed protographs is provided based on the work in \cite{Huang}. To check the validity of the proposed QC LDPC codes, we show that the proposed QC LDPC codes have larger upper bounds on the minimum Hamming distance than the QC LDPC codes lifted from single-edge protographs. Finally, the error correcting performance of the proposed QC LDPC codes is compared with those of PEG LDPC codes and PEG QC LDPC codes via numerical analysis.

\end{document}